\documentclass[draftclsnofoot,onecolumn]{IEEEtran}
\usepackage{subfigure}
\usepackage{mdwlist}
\usepackage{cite}
\usepackage{amsfonts}
\usepackage{amssymb}
\usepackage[cmex10, fleqn]{amsmath}
\usepackage{color}
\usepackage{algorithm}
\usepackage{algorithmic}
\usepackage{multirow}
\usepackage{bm}
\usepackage{booktabs}
\usepackage{graphics}
\usepackage{graphicx}
\usepackage{ifpdf}
\usepackage{xcolor}

\usepackage[T1]{fontenc}
\usepackage{aecompl}

\usepackage{float}
\usepackage{epsfig}
\usepackage{multirow}
\usepackage{enumerate}

\usepackage{nomencl}
\makenomenclature

\usepackage{arydshln}

\usepackage{tikz,times}

\usepackage{caption}

\usepackage[T1]{fontenc}
\usepackage{aecompl}

\usetikzlibrary{calc, shapes, backgrounds}

\newtheorem{theorem}{Theorem}
\newtheorem{lemma}{Lemma}

\newtheorem{proposition}{Proposition}
\newtheorem{proof}{Proof}

\ifCLASSINFOpdf
\else
\fi
\IEEEoverridecommandlockouts
\begin{document}

%
\title{ Distributed Load Shedding for Microgrid with Compensation Support via Wireless Network }
%

\author{\IEEEauthorblockN{~Qimin~Xu,~Bo~Yang,~Cailian~Chen,~Feilong~Lin,~Xinping~Guan }

\IEEEauthorblockA{Department of Automation, Shanghai Jiao Tong University, Shanghai, China\\
Collaborative Innovation Center for Advanced Ship and Deep-Sea Exploration, Shanghai, China\\
Key Laboratory of System Control and Information Processing, Ministry of Education of China, Shanghai, China} \\
Email: \{qiminxu, bo.yang, cailianchen, bruce\_lin, xpguan\}@sjtu.edu.cn}


%


\maketitle

\begin{abstract}
Due to the limited generation and finite inertia, microgrids suffer from a large frequency and voltage deviation which can lead to system collapse. {Thus, reliable load shedding method is required to maintain the frequency stability.} Wireless network, benefiting from the high flexibility and low deployment cost, is considered as a promising technology for fine-grained management. In this paper, a distributed load shedding solution via wireless network is proposed for balancing the supply-demand and reducing the load-shedding amount.
Firstly, real-power coordination of different priority loads is formulated as an optimisation problem. { To solve this problem, a distributed load shedding algorithm based on subgradient method (DLSS) is developed for gradually shedding loads.}
Using this method, power compensation can be utilised and has more time to decrease the power deficit, consequently reducing the load-shedding amount.
{Secondly, a multicast metropolis schedule based on TDMA (MMST) is developed.
In this protocol, time slots are dedicatedly allocated to increase the response rate.
A checking and retransmission mechanism is utilised to enhance the reliability of our method.}
Finally, the proposed solution is evaluated by NS3-Matlab co-simulator. The numerical results demonstrate the feasibility and effectiveness of our solution.\end{abstract}
\IEEEpeerreviewmaketitle

\section{Introduction}

Due to the depleting fossil fuel resources, rising energy costs, and deteriorating environmental conditions, more distributed energy resource (DER) units are incorporated into the current electrical power system. Microgrids are developed to interconnect the DER units in a relatively small area.
However, there exist several technical challenges in integrating DER units due to the nature of microgrids, such as limited generation, finite inertia and distributed structure.
Thus, determining how to monitor and manage the numerous DER units and loads is a critical issue, especially when the load and generation drastically change or faults happen. In this context, restoration is the typical operation to keep the supply-demand balance of the system by load shedding or generator power regulation.

Various kinds of restoration methods have been proposed to shed the appropriate loads using different methodologies \cite{Laghari2013Application,Gao2016Dynamic,hong2013multiscenario}. In \cite{Gao2016Dynamic} and \cite{hong2013multiscenario}, centralised methods were designed to coordinate multiple generators and loads in a microgrid. However, centralised methods need the collection of global information, and they easily suffer from single point failure.
Besides, centralised methods are ill suited to the structural nature of microgrids.
Thus, {distributed methods have been developed to address the above problems. }
Multi-agent system (MAS) based methods were proposed for reliable load shedding of microgrids \cite{Distributed2014Lim} and restoration of the microgrid in all-electric ship \cite{Multiagent2007Huang}.
These two algorithms were designed for the power system with specific structures.
Additionally, the restoration decision requires sophisticated coordination and information exchange between different agents, while the convergence and stability of the proposed algorithms have not been rigorously analysed.
To overcome these shortcomings, consensus based methods were applied to this problem \cite{zhao2013optimal,liang2012multiagent,xu2011stable,liu2015improved,gu2014adaptive}.
In \cite{zhao2013optimal}, an optimal load control scheme is designed based on power system model to reduce the mismatch between load and generation which is caused by sudden generation drop.
In \cite{liang2012multiagent,xu2011stable}, global information discovery (GID) algorithms for load shedding were proposed based on different consensus methods.
In \cite{liu2015improved}, a two-layer improved average consensus algorithm was designed for load shedding, which took cost and marginal cost into considerations.
In \cite{gu2014adaptive}, a decentralised under frequency load shedding (UFLS) was implemented based on the global information.
{These two works evaluated power deficiency by the rate of change of frequency (ROCOF) only at first frequency threshold which is a semi-adaptive scheme.}
However, they only shed the corresponding load amount, without consideration of mitigating the impact of load shedding on customer's experience. {The high pervasive smart meters and appliance with automatically sense and control function can be available in the future \cite{taneja2012defining}. Thus, more fine-grain load management can be realised by the collaboration of smart homes/buildings and worth further investigation.}

The conventional method for load shedding based on the ROCOF can estimate power deficit, but cannot obtain more load information such as load priority and economy. Thus, for more fine-grain load management, utilising advanced information and communication technology is necessary.
In the former works \cite{zhao2013optimal,xu2011stable,liu2015improved,gu2014adaptive}, the ideal communication model was employed. However, as for the load shedding operation performance, it is not only determined by the control algorithm but also related to the protocol design of communication system \cite{Parikh2010Power}.
{Compared with wireline networks, wireless networks bring the benefits of high flexibility, low-cost deployment, and widespread access, which are suitable for microgrids with numerous distributed DER units and loads.}
Thus, wireless networks, such as wireless LAN and LTE, are potential technologies to realise intelligent management.
The round-robin polling mechanism based time-division multiple access (TDMA) is a considerable protocol for wireless access in microgrids \cite{yang2011communication}.
For distributed coordination and fast convergence, several protocols were designed based on a unicast mode to coordinate agents in microgrids \cite{liang2012multiagent}. These protocols are only proposed for the GID. Additionally, the packet loss and hidden terminal problem are not taken into consideration.
Therefore, since load shedding method is time-sensitive, the protocol considering time efficiency and transmission reliability is urgently needed.

{In this paper, a fully distributed load shedding solution is proposed, which aims to improve customer's experience by fine-grain load management with reliable communication protocol.}
The main contributions of this paper are as follows:
\begin{itemize}
\item{}
Considering that the distributed management of small-scale microgrid contains loads with different priorities, the real power of loads are coordinated by utilisation level. Thus, a load priority associated optimisation problem that aims at maximising the weighted sum of the remained loads and balancing the supply and demand is formulated, which has a non-smooth objective.

\item{}
For reducing the impact of load shedding on customer's experience, a DLSS method is proposed to shed loads gradually by the frequency deviation rather than at a fixed number of steps and fixed load-shedding amount. Hence, power compensation can be utilised to reduce the load-shedding amount. Moreover, the relevant analysis of convergence is presented.

\item{}
A multicast metropolis schedule based on TDMA (MMST) is developed to increase the response rate and guarantee the reliability of DLSS method. In this protocol, a time slot allocation algorithm is designed to increase the number of concurrent transmission between agents, and a data frame structure piggybacks the checking information of packets received from neighbour agents.
\end{itemize}

The paper is organized as follows. In Section \ref{sec:system_structure}, the system structure is introduced. Section \ref{sec:load_shedding_solution} presents in detail the distributed load shedding solution. In Section \ref{sec:protocol_design}, the proposed MMST protocol is elaborated.
The performance of the proposed solution is evaluated and compared with the existing methods in Section \ref{sec:simulation}. Finally, the conclusion is drawn in Section \ref{sec:conclusion}.

\section{System Structure}
\label{sec:system_structure}

In this research, we consider a load shedding problem in a microgrid.
The microgrid network is denoted by a graph $ ( \mathcal{N}, \mathcal{E}) $. $\mathcal {N}$ denotes the bus set which is defined as $\mathcal{N} = \{1, \cdots , N \} =  \mathcal{N}_{cg} \cup \mathcal{N}_{ig} $. $ \mathcal{N}_{cg} $ and $ \mathcal{N}_{ig} $ are the bus sets connected with conventional distributed generators (DG) and inverter-based DGs respectively.
$\mathcal{E} \subseteq \mathcal{N} \times \mathcal{N}$ denotes the set of transmission line interconnecting the buses. In the microgird, there is at least one bus connected with synchronous generator (SG) or energy storage system (ESS), which can be used for power compensation.

\textbf{Assumption:} We make the following assumptions:
\begin{itemize}
\item Each bus is managed by an agent which is regarded as a regional controller. The agents communicate with each other via wireless networks.
\item Each agent has the information of global maximum generation {  capacity. But they do not have the real-time power generation and load demand of other buses.}
\item The communication range of each agent only covers its neighbour agents, since the communication network of microgrids has low density.
\end{itemize}

{  \subsection{Generation Model and Multi-Priority Load Model}
\label{sec:generation_and_demand}
 }

The total power generation $P_G$ in the microgrid can be obtained by
\begin{align}
\label{eqn:generation_model}
P_{G} = \sum_{i=1}^{N} P_{G_i},
\end{align}
where $ P_{G_i} $ denotes the power generation at bus $i$, $N$ denotes the number of buses (agents) in the microgrid.

The total real power demand $P_D$ and load model \cite{New2015Laghari} can be expressed as
\begin{align}
\label{eqn:demand_model}
& P_{D} = \sum_{i=1}^{N} P_{L_{i}} + P_{loss}  = \sum_{i=1}^{N} \sum_{l=1}^{N_{L,i}} b_{i,l} P_{L_{i,l}} + P_{loss}, \\
\label{eqn:load_model}
& P_{L_{i,l}} = P_{L_{i,l}}(0) (1 + \kappa_f \Delta f + \kappa_v \Delta V),
\end{align}
where $P_{L_{i,l}}(0)$ denote the real power of load $l$ at base frequency and voltage, and $P_{L_{i,l}}$ at new voltage and frequency. $\Delta f$ and $\Delta V$ denote the deviation of system frequency and voltage, respectively. $\kappa_f$ and $\kappa_v$ are the coefficients of real power load dependency on frequency and voltage respectively. $b_{i,l}$ represents the control variable of load $l$ at bus $i$. Thus, $\bm b_i$ is an array of length $N_{L,i}$ (1/0 = active/non-active), where $N_{L,i}$ denotes the number of loads at bus $i$.
$P_{loss}$ is the real power loss in transmission line.

The real power of the total loads $P_L^{\max}$ {when faults happen} can be calculated as
\begin{align}
 P_{L}^{\max} = \sum_{i=1}^{N} P_{L_i}^{\max} = \sum_{i=1}^{N} \sum_{l=1}^{N_{L,i}} P_{L_{i,l}}(0),
\end{align}
where $P_{L_i}^{\max}$ denotes the total power of the loads at bus $i$ when they are all in active status.

The loads are divided into $G$ grades according to the economic and social influence caused by load interruption. $G$ denotes the maximum load grade, and $G =3 $ in most cases. The vital load is the uninterruptible power-supplied load, which would cause great economic losses, and even casualty if interrupted. The second grade load would cause certain economic losses if interrupted. The nonvital load is the third grade load which can be adjusted. Hence $ P_{L}^{\max}$ can also be represented by
\begin{align}
\label{eqn:grade_load}
P_{L}^{\max} = \sum_{g=1}^{G} \rho_g P_{{L}}^{\max} = \sum_{i=1}^{N} \sum_{g=1}^{G} \rho_{g,i} P_{L_{i}}^{\max},
\end{align}
where
$\rho_g$ is the ratio of the $g$-th grade loads in the real power of the total loads $P_L^{\max}$, and $\rho_{g,i}$ denotes the ratio of the $g$-th grade loads in the total real power at bus $i$ .

For different priority loads, {$w_{g}$ }denotes the weight factor of the $g$-th loads, which is used to set a measurable indicator of load shedding and ensure that the lower priority loads are shed first.
The smaller $g$ is, the higher priority the loads have.
{Thus, according to (\ref{eqn:grade_load}), the weighted sum of all the load power is written as
}
\begin{equation}
\label{eqn:total_weight}
{P_{W_t}} = \sum_{g=1}^G w_{g} \rho_g P_{L}^{\max},
\end{equation}
where $w_{g}$ decreases with the increase of load priority.

{For the load shedding problem, $P_{L_i}$ is the adjustable variable, which denotes the remained power of loads at bus $i$.
The utilization level $u_i$ is used to coordinate power of loads at each bus, which is defined as}
\begin{equation}
\label{eqn:utilization_level}
\begin{aligned}
u_i & =\dfrac{ P_{L_i}} {P_{L_i}^{\max}} = \dfrac{\sum_{l=1}^{N_{L,i}} b_{i,l} P_{L_{i,l}}(0)} {P_{L_i}^{\max}} \\
& + \dfrac{ \sum_{l=1}^{N_{L,i}} b_{i,l} ({\kappa_f \Delta f + \kappa_v \Delta V}) P_{L_{i,l}}(0) } {P_{L_i}^{\max}} , \ P_{L_i} \in [0,P_{L_i}^{\max} ].
\end{aligned}
\end{equation}

{  \subsection{ Power Deficit and Load-Shedding Amount Formulation} }

The power deficit $\Delta \tilde{P}$ can be estimated based on the ROCOF.
If initial power deficit $ \Delta \tilde{P} $ caused by the fault is in the range as follow
\begin{equation}
\label{eqn:shedding_load_range}
\sum_{g=m+1}^G \rho_g P_{L}^{\max} < \Delta \tilde{P} \leqslant \sum_{g=m}^G \rho_g P_{L}^{\max},
\end{equation}
the corresponding total weighted sum of load power that need to be shed $P_{W_{\Delta}}$ can be expressed as
\begin{equation}
\label{eqn:delta_weight}
\begin{aligned}
{P_{W_{\Delta}}} & =  \sum_{g=m+1}^{G}  w_{g} \rho_g P_{L}^{\max}    + w_{{m}} \left(\Delta \tilde P - \sum_{g=m+1}^G \rho_g P_{L}^{\max} \right).
\end{aligned}
\end{equation}

Similarly, the weighted sum of the remained load power $P_{W_{i}}$ at bus $i$ based on (\ref{eqn:grade_load}) and (\ref{eqn:utilization_level}) can be calculated as
\begin{equation}
\begin{aligned}
\label{eqn:each_weight}
P_{W_{i}} =
& \sum_{g=1}^{m} w_{g}\rho_{g,i} P_{L_{i}}^{\max}  + w_{{m+1}} \left(u_i - \sum_{g=1}^m \rho_{g,i} \right) P_{L_{i}}^{\max}, \\
& \text{if} \ u_i
\in \left( \sum_{g=1}^{m} \rho_{g,i}, \ \sum_{g=1}^{m+1} \rho_{g,i} \right].
\end{aligned}
\end{equation}

The left side of Fig. \ref{fig:load_diagram} shows the relationship between different priority loads, $P_{W_{t}}$ and $P_{W_{\Delta}}$
in (\ref{eqn:total_weight}) and (\ref{eqn:delta_weight}). Part \textcircled{1} and part \textcircled{2} represent the two terms of (\ref{eqn:delta_weight}) respectively. The right side illustrates the relationship between different priority loads and utilization level $u_i$ at bus $i$ in (\ref{eqn:each_weight}). Part \textcircled{3} and part \textcircled{4} represent the two terms of (\ref{eqn:each_weight}) respectively.
{Thus, the objective of load shedding in this work is to satisfy
\begin{equation}
\label{eqn:load_balance}
 P_{W_t} - P_{W_{\Delta}} = \sum_{i=1}^N P_{W_i}
\end{equation}
where the right term is estimated based on system information and ROCOF, and the left is adjust variable.}
\begin{figure}[ht]
\vspace{-2ex}
\begin{center}
\setlength{\belowcaptionskip}{-0.2cm}
\includegraphics[width = 0.46 \textwidth]{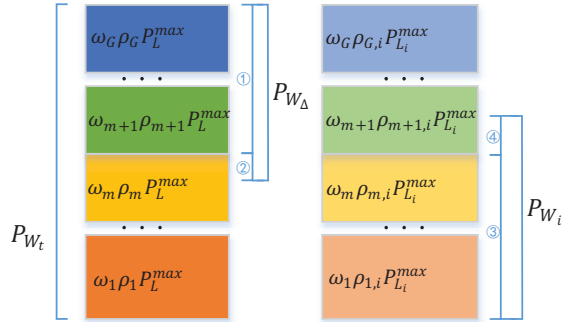}
\caption{{ Multi-Priority Load Diagram.}}
\captionsetup{justification=centering}
\label{fig:load_diagram}
\end{center}
\vspace{-2ex}
\end{figure}


\section{Distributed Load Shedding Solution }
\label{sec:load_shedding_solution}

In this section, the distributed load shedding solution is introduced, which is shown in Fig. \ref{fig:solution_figure}. When fault causes overload, such as islanding and generation loss, the system starts load shedding process. Due to the load shedding method depending on the operating information of the microgrid, a GID is executed in the first stage.
In normal condition, this operation runs periodically. The GID algorithm and its adopted communication protocol determine the minimum convergence time $T_{gi}$.
Once the system frequency is lower than the trigger frequency $f_{tr}$, a GID process is executed.
{When the global information is obtained, the DLSS method is carried out at each agent after a time delay $t_{ad}$, which disconnects loads gradually with consideration of load priority.} This process is ended when the power balance is achieved, which consumes time $T_{ls}$. Considering that the frequency may drop to the unsafe range before the convergence of DLSS process, a safety threshold shedding is utilised.The proposed solution is detailed as follow.

\begin{figure}[htbp]
\centering
\vspace{-10pt}
\subfigure[]{\label{fig:solution_process}
\includegraphics[width=0.42 \textwidth]{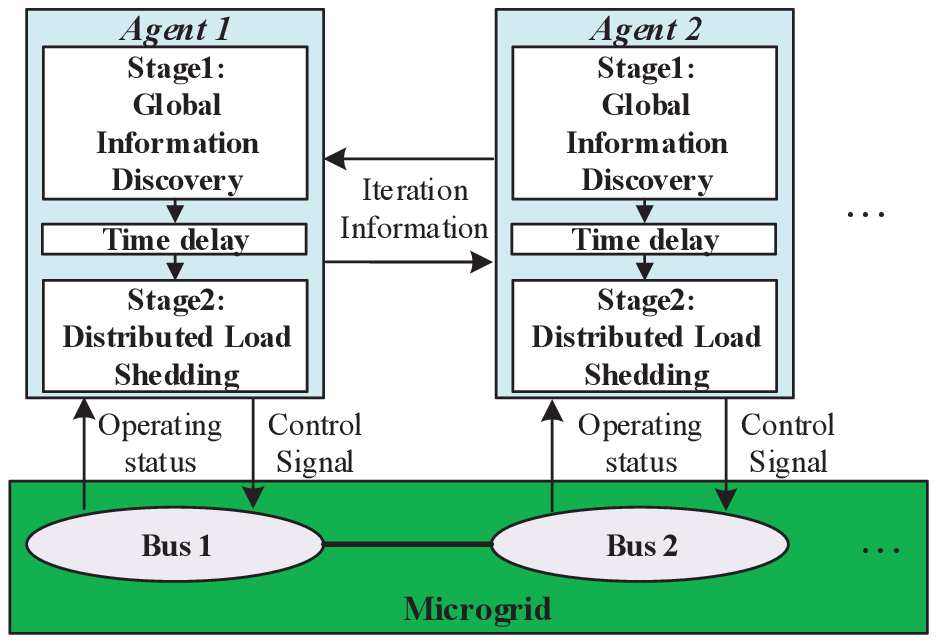}}
\subfigure[]{\label{fig:frequency_regulation}
\includegraphics[width = 0.42 \textwidth]{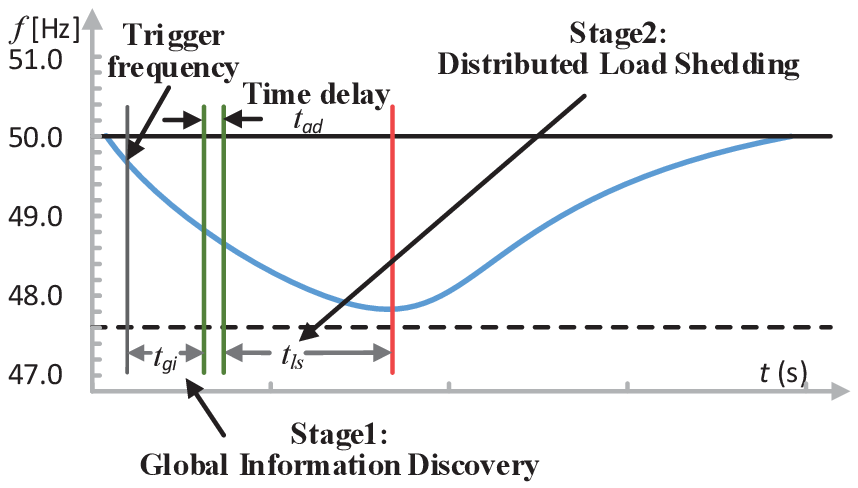}}
\caption{ {(a) Distributed load shedding operation process. (b) Frequency regulation.} }
\label{fig:solution_figure}
\end{figure}

\subsection{Global Information Discovery}
\label{subsec:global_info_discovery}

In this research, the global information of the microgrid that needs to be discovered includes three types: load information $P_L$, $\rho_g$, system information $f$, $V$, and power deficit $\Delta P$. Agents obtain the global information ${X}$ by average consensus method which only needs that each agent exchanges data with directly connected agents. This method can improve the estimation perfermance by reducing the measurement noise and oscillation of frequency \cite{zhao2013optimal}. The average consensus of local information $x_i$ will converge to the common value $ \bar{X} $ which is expressed as
\begin{equation}
\label{eqn:average_consensus}
\begin{aligned}
& \bar{X} = \dfrac{1}{N} \sum_{i \in \mathcal{N}} {x_i^{(0)}} \\
& {X} = N \bar{X}.
\end{aligned}
\end{equation}
{where $x_i^{(0)}$ is the inital value of local information $x_i$.}

Estimation of power deficit $\Delta P$ is the key point to determine the magnitude of shedding loads in the microgrid.
Due to the two types of generators, the estimations of $\Delta P_i$ are different.
$ \Delta P_i $ denotes the real power deficit at bus $i$.

{\textbf{Conventional DGs:}} For conventional DGs without inverters, the magnitude of the power deficit can be estimated by
\begin{equation}
\label{eqn:dg_rocof}
\Delta P_i =  \dfrac{2 H_{cg,i}}{f_{no}} \dfrac{d f_{cg,i}}{d t}, \ i \in \mathcal{N}_{cg},
\end{equation}
where the inertia constant $H_{cg,i}$ of conventional DGs are deterministic. The ROCOF ${d f_{cg,i}}/{d t}$ is measured when the imbalance of the real power occurs. $f_{no}$ is the base frequency.

{\textbf{Inverter-based DGs:}} The energy sources that connected to the microgrid system with inverters have little contribution to the system inertia, such as photovoltaics (PV). Hence, the power deficit of the inverter-based DG is estimated by droop control characteristic. The relationship of real power and frequency is similar to the conventional DGs. The magnitude of the power deficit of the inverter-based DG can be calculated by:
\begin{equation}
\label{eqn:idg_rocof}
\Delta P_{i} = \dfrac{2 \pi (f_{ig,i} - f_{no})}{\xi_i} = \dfrac{2 \pi \Delta f_{{ig},i}}{\xi_i}, \ i \in \mathcal{N}_{ig},
\end{equation}
where $\Delta f_{ig,i}$ denotes the measured frequency deviation of inverter-based DG $i$
and $\xi_i$ is the droop coefficient.

{ $T_{gi}$ determines the response rate of load shedding process according to Fig. \ref{fig:solution_figure}.} To minimize the $T_{gi}$, the MMST protocol is designed, which is introduced in section \ref{sec:protocol_design}.

\subsection{ Distributed Load Shedding Algorithm Based on Subgradient Method}
\label{sec:shedding_solution}

{Generator compensation is utilised with the load shedding process to reduce the load-shedding amount.
Hence, the real power deficit is decreased by load shedding and generator compensation together. Generator gradually increases the real power to compensate the power deficiency.
However, the speed of the generator compensation depends on the generating unit type and is usually unchangeable. Thus, prolonging the time to the unsafe range by gradually load shedding can give more time for the generator compensation.
The more power the generators compensate, the fewer loads the system disconnects. Therefore, a DLSS method is proposed for collaborative operation with generator compensation.}

For better quality of customer experience, the objective is to maximise the weighted sum of loads.
To make the problem tractable, the objective is transformed into minimizing deviation between the current weighted sum of loads $\sum_{i=1}^N P_{W_i}$ and the prediction weighted sum of loads after compensation $(P_{W_t} - P_{W_{\Delta}})$ according to (\ref{eqn:load_balance}).
Since the balance between power supply and demand is the basic requirement, the problem is formulated as follows
\begin{subequations}
\label{eqn:objective_constraint}
\begin{align}
\label{eqn:objective} \underset{\bm{u} } {\mathrm{min}} \quad & F( \bm{u} )= \left(  P_{W_t} - P_{W_{\Delta}} - \sum_{i=1}^N P_{W_i} \right)^2 \\
\label{eqn:constraint} \text{s.t.} \quad & \left( \sum_{i=1}^N (u_i P_{L_i}^{\max} )+ P_{loss} - P_{G} \right)^2 \leqslant \varepsilon, \\
\label{eqn:gen_constraint} & P_{G,i}^{\min} \leqslant P_{G,i} \leqslant P_{G,i}^{\max}, \\
 & {(\ref{eqn:generation_model})-(\ref{eqn:utilization_level}), (\ref{eqn:delta_weight}), (\ref{eqn:each_weight}), (\ref{eqn:dg_rocof}), (\ref{eqn:idg_rocof}) } \nonumber
\end{align}
\end{subequations}
where $\bm{u} = (u_1, \cdots, u_N)^{T} \in \mathcal{U} $, $\mathcal{U} $ the utilization level set, (\ref{eqn:constraint}) the power balance constraint between supply and demand, and $\varepsilon$ is the maximum error between power supply and demand.

Due to tens or hundreds of loads at each bus, the interval between adjoining points of $u_i$ is small to be around or below to one percent. Thus, $u_i$ can be linearised although it is a discrete variable.
It is noted that problem (\ref{eqn:objective_constraint}) is a convex optimization problem. In this paper, a distributed load shedding algorithm based on subgradient method is used to solve it, which is referred to\cite{xu2014distributed,chang2014distributed}.
We consider the Lagrange dual problem of (\ref{eqn:objective_constraint}):
\begin{equation}
\label{eqn:lagrange_dual}
\underset{ {\lambda} \geqslant 0} {\max} \left\{ \underset{\bm{u} } {\min \mathcal{L}} (\bm{u}, {\lambda}) \right\} ,
\end{equation}
where ${\lambda} $ is the dual variable associated with the inequality constraint (\ref{eqn:constraint}), which is non-negative. The Lagrange function is expressed as:
\begin{equation}
\label{eqn:lagrange_function}
 \mathcal{L} (\bm{u}, {\lambda}) = F \left(u_1, \cdots, u_N \right) + {\lambda}  J \left( u_1, \cdots, u_N  \right),
\end{equation}
where $J(\bm{u}) = ( \sum_{i=1}^N (u_i P_{L_i}^{\max}) + P_{loss} - P_{G} )^2 - \varepsilon $. $\mathcal{L}_{\bm{u}} \big(\bm{u}^{(k)}, {\lambda}^{(k)} \big)$ and $\mathcal{L}_{\lambda} \big(\bm{u}^{(k)}, {\lambda}^{(k)} \big)$ represent the subgradients of $ \mathcal{L} $ at $(\bm{u}^{(k)}, \lambda^{(k)})$ with respect to $\bm{u}$ and $ \lambda $ respectively, which are given by
\begin{equation}
\begin{aligned}
\label{eqn:largrange_u} & \mathcal{L}_{\bm{u}} \left(\bm{u}^{(k)}, {\lambda}^{(k)} \right) & \\
 & =
\left[ \begin{array}{c}
\mathcal{L}_{u_1} \Big(\bm{u}^{(k)}, {\lambda}^{(k)} \Big)  \\
\vdots\\
\mathcal{L}_{u_N} \Big(\bm{u}^{(k)}, {\lambda}^{(k)} \Big)  \\
\end{array} \right] & \\
&  =
\left[ \begin{array}{c}
{\nabla F(u_1^{(k)})} + \lambda^{(k)} {\nabla J(u_1^{(k)})} \\
\vdots\\
{\nabla F(u_N^{(k)})} + \lambda^{(k)} {\nabla J(u_N^{(k)})}
\end{array} \right],
\end{aligned}
\end{equation}
\begin{equation}
\begin{aligned}
\label{eqn:largrange_lambda} & \mathcal{L}_{{\lambda}} \left(\bm{\bm{u}}^{(k)}, {\lambda}^{(k)} \right) \\
& = \left( \sum_{i=1}^N \left( u_i^{(k)} P_{L_i}^{\max} \right) + P_{loss} - P_{G} \right)^2  - \varepsilon .
\end{aligned}
\end{equation}

Frequency and voltage deviation affect the consumed power of loads based on (\ref{eqn:utilization_level}). In addition, the voltage fluctuates in the load shedding process. Thus, the utilisation level for determining grade of shedding loads $u_i(f)$ takes into account of frequency deviation, which is written as
\begin{equation}
u_i(f) = u_i - \dfrac{ \sum_{l=1}^{N_{L,i}} b_{i,l} ({\kappa_f \Delta f}) } {P_{L_i}^{\max}}.
\end{equation}

Since $ {F( u_i^{(k)} )}$ is non-smooth, there are two subgradient at point $u_i = \sum_{g=1}^{m} \rho_{g,i}, 1 \leqslant m < G$. Hence the $ {\nabla F( u_i^{(k)}  )}$ is calculated by using (\ref{eqn:each_weight}) and (\ref{eqn:objective}), i.e.,
\begin{equation}
\begin{aligned}
\label{eqn:pd_W} {\nabla F \left( u_i^{(k)} \right) } & = -2P_{L_{i}}^{\max} w_{{m+1}} \left( P_{W_t} - P_{W_{\Delta}}- \sum_{i=1}^N P_{W_i}^{(k)} \right),  \\
& \text{if} \ u_i(f)
\in \left( \sum_{g=1}^{m} \rho_{g,i}, \ \sum_{g=1}^{m+1} \rho_{g,i} \right].
\end{aligned}
\end{equation}

The $ {\nabla J( u_i^{(k)} )} $  can be calculated by using (\ref{eqn:constraint}), and they can be expressed as
\begin{equation}
\begin{aligned}
\label{eqn:pd_P} {\nabla J( u_i^{(k)}  )}  & =  2P_{L_{i}}^{\max} \left(  P_{loss} + \sum_{i=1}^N u_i^{(k)}  P_{L_i}^{\max} - P_G  \right) .
\end{aligned}
\end{equation}

The four steps of the proposed DLSS is described in detail as follows:
\subsubsection{Global variable estimation}
{Each agent makes load shedding decision based on the global information $P_{W_i}$ and $\Delta P_i$, which cannot be obtained directly. Thus, two auxiliary variables denoted by ${W_i}^{(k)}$ and ${D_i}^{(k)}$ are added to estimate them respectively. $k$ is the updating index.}
The two auxiliary variables represent respectively the average estimates of the weighted sum of loads $1/N \sum_{i=1}^{N} {P_{W_i}}^{(k)} $ and of the power demand $1/N \left(  P_{loss} + \sum_{i=1}^N u_i^{(k)}  P_{L_i}^{\max} - P_G  \right) = 1/N \sum_{i=1}^N \Delta P_i^{(k)} $.
Due to the distributed feature of our algorithm, each agent has a copy of the dual variable $\lambda_i^{(k)}$ instead of $\lambda^{(k)}$.
$ u_i^{(k)}$ is the utilization level of agent $i$ at updating index $k$.
Each agent $i$ sends ${W_i}^{(k-1)}$, ${D_i}^{(k-1)}$, $ \lambda_i^{(k-1)}$, and $u_i^{(k-1)}$ to all the neighbour agents $j$ satisfying $ j \in \mathcal{N}_i$. $\mathcal{N}_i$ denotes the neighbour set of agent $i$. Each agent $i$ also receives ${W_i}^{(k-1)}$, ${D_i}^{(k-1)}$, $ \lambda_i^{(k-1)}$, and $u_i^{(k-1)}$ from its neighbour agents, and estimates the global variable based on the those data
\begin{subequations}
\label{eqn:local_state_update}
\begin{align}
& \tilde W_i^{(k)} = \sum_{j=1}^{N} {a_{ij} {W_j}^{(k-1)}},  \tilde D_i^{(k)} = \sum_{j=1}^{N} {a_{ij}} {D_j}^{(k-1)}, \\
& \tilde{\lambda}_{i}^{(k)} = \sum_{j=1}^{N} {a_{ij}} {\lambda}_j^{(k-1)},
\label{eqn:P_S_update1} \tilde u_{i}^{(k)} = \sum_{j=1}^{N} {a_{ij}} u_{j}^{(k-1)},
\end{align}
\end{subequations}
where $a_{ij}$ is the information exchange coefficient between agent $i$ and $j$. When $(i,j) \in \mathcal{E}$, $a_{ij}>0$ holds and $a_{ij}=0$ otherwise.
The $n$ dimensional transition matrix $A$ is composed of $a_{ij}$s. The transition matrix $A$ is a doubly stochastic matrix, which satisfies that $\sum_{j=1}^{N} a_{ij} = 1$ for all $i$ and $\sum_{i=1}^{N} a_{ij} = 1$ for all $j$. {$\tilde u_{i}^{(k)}$ denotes the estimated global utilization level of loads at agent $i$, which is used for load-shedding amount correction.}

\subsubsection{Primal-dual variable update}
Because function $F{(\bm{u})}$ is no-smooth, each agent $i$ updates its primal and dual variables $ (u_i^{(k)}, \lambda_i^{(k)})$ based on the estimated global variable $( \tilde W_i^{(k)}, \tilde D_i^{(k)}, \tilde{\lambda}_{i}^{(k)})$ as follow:
\begin{equation}
\begin{aligned}
\label{eqn:u_varible_update} {u_i}^{(k)} & = \left( {u_i}^{(k-1)} -  \tau_k \mathcal{L}_{u_i} \left( \bm{u}^{(k-1)}, \tilde {\lambda}_i^{(k)} \right) \right)^+ \\
& = \left( {u_i}^{(k-1)} -  2 \tau_k P_{L_{i}}^{\max} \left( {\lambda}_i^{(k)} N \tilde D_i^{(k)} \right. \right. \\
& \left. \left. \quad - w_{{m+1}} \big( P_{W_t} - P_{W_{\Delta}} - N \tilde W_i^{(k)} \big)  \right) \right)^+, \\
& \quad u_i(f)
\in  \left( \sum_{g=1}^{m} \rho_{g,i}, \ \sum_{g=1}^{m+1} \rho_{g,i} \right],
\end{aligned}
\end{equation}
\begin{equation}
\begin{aligned}
\label{eqn:lambda_varible_update} {\lambda_i}^{(k)}
= & \left( \tilde {\lambda}_i^{(k)}  + \tau_k \mathcal{L}_{{\lambda_i}} \left(\bm{\bm{u}}^{(k)}, {\tilde \lambda_i}^{(k)} \right) \right)^{+} \\
= & \Big( \tilde {\lambda}_i^{(k)}  + \tau_k \Big( \Big( N \tilde D_i^{(k)} \Big)^2  - \varepsilon \Big)^{+},
\end{aligned}
\end{equation}
where $\tau_k$ is the step size, $(u)^{+} = \max \{ u,0 \}$.

The shedding sequence of loads is sorted in an ascending order of weighted power $w_{g} P_{L_{i,l}}$. Then, the shedding control variable $\bm b_i$ can be determined by approximating the obtained $u_i$ based on (\ref{eqn:utilization_level}).


\subsubsection{Local variable update}
{When the load shedding decision is carried out, $W_i^{(k)}$ and $D_i^{(k)}$ need to be updated for the global information estimation in next iteration.}
Each agent $i$ updates variable ${W_i}^{(k)}$ and ${D_i}^{(k)}$ with the changes of the local argument functions $P_{W_i}^{(k)}$ and ${\Delta P_i}^{(k)}$,
\begin{align}
\label{eqn:auxiliary_1_update} & {W_i}^{(k)} = \tilde W_i^{(k)} + P_{W_i}^{(k)} - P_{W_i}^{(k-1)}, \\
\label{eqn:auxiliary_2_update} & {D_i}^{(k)} = \tilde D_i^{(k)} + \Delta P_i^{(k)} - \Delta P_i^{(k-1)}.
\end{align}

In order to reduce the jitter of utilization level, the initial value of ${W_i}^{(0)} $ and ${D_i}^{(0)}$ can be set to $P_{W_t}/N$ and $\sum_{i=1}^{N} \Delta P_i/N$ based on the obtained data in the GID.

\subsubsection{load-shedding amount correction }
The SG or ESS is controlled to generate real power to compensate for the deficiency in the process of load shedding. The compensation rate is determined by the specification and power control algorithm of generator or ESS, which is a relatively slower than load shedding.
The total loads that need to be shed $\Delta \tilde P$ is updated as follow:
\begin{align}
\label{eqn:P_S_update2} & \Delta \tilde P_{}^{(k)} = \Delta {P}^ {(k)} + (1 - \tilde u_i^{(k)})P_L^{\max},
\end{align}
where the first term of (\ref{eqn:P_S_update2}) is the current power deficit, and the second term is the sum of loads that have been shed. Thus the weighted sum of loads that need to be shed $P_{W_{\Delta}}$ can be updated based on (\ref{eqn:total_weight}) and (\ref{eqn:P_S_update2}).
Due to this correction mechanism, a dynamic load shedding method is realised.

The above steps of the DLSS method are summarised in Algorithm \ref{alg:DLSS}.

\begin{algorithm}[htb]
  \caption{Distributed Load Shedding Algorithm Based on Subgradient Method}
  \label{alg:DLSS}
  \begin{algorithmic}[1]
  \REQUIRE Initial variables ${u}_i^{(0)}$, $\lambda_i^{(0)}$, ${W_i}^{(0)}$, ${D_i}^{(0)}$, and ${\Delta P^{(0)} }$.
  \ENSURE variables $\bm{u}$.
    \STATE Set $k = 1$.
    \REPEAT
        \STATE Exchange $W_i^{(k)}$ and $D_i^{(k)}$ with its neighbour agents;
        \STATE Estimate the average local variables ${\tilde W_i}^{(k)} $, ${\tilde D_i}^{(k)}$ and $\tilde \lambda_i^{(k)} $ by (\ref{eqn:local_state_update});
        \STATE  Update primal and dual variables $ {u_i}^{(k)} $ and $ \lambda_i^{(k)} $ by (\ref{eqn:u_varible_update}) and (\ref{eqn:lambda_varible_update});
        \STATE  {Calculate $\bm b_i$ based on $u_i$ and $(\ref{eqn:utilization_level})$, and shed the corresponding loads;}
        \STATE  Update the local variables  $ {W_i}^{(k)} $ and ${ D_i}^{(k)}$  by (\ref{eqn:auxiliary_1_update}) and (\ref{eqn:auxiliary_2_update});
        \STATE  {Correct load-shedding amount $\Delta \tilde P^{(k)}$ based on (\ref{eqn:P_S_update1}) and (\ref{eqn:P_S_update2});
        }
        \STATE $k = k+1$;
    \UNTIL { Satisfy power balance constraint (\ref{eqn:constraint}) }
  \end{algorithmic}
\end{algorithm}

\subsection{ Safety threshold shedding }
In the gradual load shedding, the safety threshold of system frequency must be guaranteed. If the system frequency drops down to the lower safety threshold $f_{th}$ and the gradual load shedding has not been finished, the safety threshold shedding is executed immediately. The load-shedding amount $P_{{sh}_i}$ follows the rule according to \cite{zhong2005power} which is
\begin{equation}
\label{eqn:safety_threshole_shedding_original}
P_{{sh}_i} = \dfrac{\Delta f_i P_{L_i}}{\sum_{i \in \mathcal{N}} \Delta f_i P_{L_i}}  \Delta {\tilde P}^{(k)}
\end{equation}
where $\Delta f_i$ denotes the frequency deviation at bus $i$ compared to the base frequency. Considering the small deviation of $\Delta f_i$ between different buses in the microgrid, (\ref{eqn:safety_threshole_shedding_original}) can be written as
\begin{equation}
\label{eqn:safety_threshole_shedding}
P_{{sh}_i} =\dfrac{u_i P_{L_i}^{\max}}{ \tilde u_i P_L^{\max} }  \Delta {\tilde P}^{(k)}
\end{equation}


{To describe the relationship between the safety threshold shedding and other modules, a high-level logic overview is given in Fig. \ref{fig:logic_flow}, which also provides a complete description of the proposed solution.
After the convergence of the GID, the DLSS method is carried out after a time delay $t_{ad}$. Thus, the total time delay equals to $t_{gi} + t_{ad}$. Regardless of the convergence of the DLSS process, the safety threshold shedding is triggered if the frequency drops to $f_{th}$. }

\begin{figure}[ht]
\begin{center}
\vspace{-2ex}
\setlength{\belowcaptionskip}{-0.2cm}
\includegraphics[width = 0.42 \textwidth]{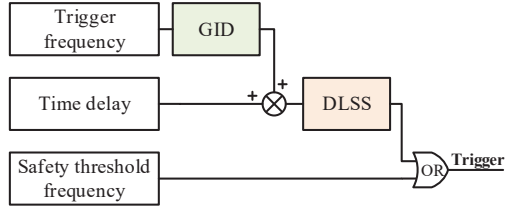}
\caption{{High level logic overview of the proposed solution.}}
\captionsetup{justification=centering}
\label{fig:logic_flow}
\end{center}
\vspace{-2ex}
\end{figure}

\subsection{Convergence Analysis of DLSS Method}

The precondition of balancing the power supply-demand is the convergence of DLSS algorithm which is analysed here.
It can be known that the convexity of function $J(\bm{u})$ implies that it has uniformly bounded subgradient, which is equivalent to $J(\bm{u})$ being Lipschitz continuous. Thus, based on (\ref{eqn:pd_W}), we have
\begin{equation}
\begin{aligned}
\label{eqn:J_condition1} \big\lVert {\nabla J( \bm{u})} \big\rVert
& =  \left\lVert \left[ \begin{array}{c}
{\nabla J(u_1^{(k)})}  \\
\vdots\\
{\nabla J(u_N^{(k)})}
\end{array} \right]
 \right\rVert \\
& \leqslant 2 \sqrt{N} \check P_{L_i} \Delta \tilde P, \; \forall  \bm{u} \in \mathcal{U} ,
\end{aligned}
\end{equation}
\begin{equation}
\begin{aligned}
\label{eqn:J_condition2}
& \big\lVert \nabla J \left(  \bm{u}\right)- \nabla J \left(  \bm{u}^{+} \right) \big\rVert \\
= &  \left\lVert \left[ \begin{array}{c}
{\nabla J(u_1^{(k)}) - \nabla J(u_1^{+(k)})}  \\
\vdots\\
{\nabla J(u_N^{(k)}) - \nabla J(u_2^{+(k)})}
\end{array} \right]
 \right\rVert \\
\leqslant & 2 \check P_{L_i}^2 \big\lVert \bm{u} - \bm{u}^{+} \big\rVert, \; \forall  \bm{u}, \bm{u}^{+} \in \mathcal{U},
\end{aligned}
\end{equation}
where $ \check P_{L_i} = \underset{ 1 \leqslant i \leqslant n} \max \{P_{L_i}^{\max}\}$.

Similarly, the convexity of function $F(\bm{u})$ implies that it has uniformly bounded subgradient. Based on (\ref{eqn:pd_P}), we have
\begin{equation}
\begin{aligned}
\label{eqn:F_condition1} \big\lVert {\nabla F( \bm{u})} \big\rVert
 & = \left\lVert \left[ \begin{array}{c}
{\nabla F(u_1^{(k)})}  \\
\vdots\\
{\nabla F(u_N^{(k)})}
\end{array} \right]
 \right\rVert \\
& \leqslant \ 2 \sqrt N \check P_{L_i} w_{{G}} P_{W_\Delta}, \forall  \bm{u} \in \mathcal{U} ,\\
\end{aligned}
\end{equation}
\begin{equation}
\begin{aligned}
\label{eqn:F_condition2}
& \big\lVert \nabla F \left(  \bm{u}\right)- \nabla F \left(  \bm{u}^{+} \right) \big\rVert \\
= & \left\lVert \left[ \begin{array}{c}
{\nabla F(u_1^{+(k)}) - \nabla F(u_1^{+(k)}) }  \\
\vdots\\
{\nabla F(u_N^{+(k)}) - \nabla F(u_N^{+(k)}) }
\end{array} \right]
 \right\rVert \\
\leqslant & \ 2 \check P_{L_i}^2 w_{{G}}^2 \big\lVert \bm{u} - \bm{u}^{+} \big\rVert, \; \forall  \bm{u}, \bm{u}^{+} \in \mathcal{U}. \\
\end{aligned}
\end{equation}

\begin{proposition}
Assume that the step size sequence ${ \tau_k }$ is non-increasing such that $\tau_k >0$ for all $ k \geqslant 1$.
Then, ${\bm u^{(k)}} $ and $ {\lambda_i^{(k)}},i = 1, \cdots, n,$ generated by the algorithm of DLSS can converge to an optimal primal solution with compensation $ \bar {\bm u}^{*} \in \mathcal U$ and an optimal dual solution with compensation $ \bar {\lambda}^*$, respectively. $C_\lambda$ denotes the upper bound of $\lVert \lambda \rVert$.
\end{proposition}

\begin{proof}
Please see the Appendix.
\end{proof}

\subsection{Parameter Setting}
\label{subsec:para_setting}

{Frequency setting and time delay are two important parameters in ROCOF relay. From the results in \cite{Efficient2006Vieira,Development2015motter}, the frequency setting is an important trigger condition for the load shedding process, which includes over-frequency setting and under-frequency setting.

From the equations in (\ref{eqn:dg_rocof}) and (\ref{eqn:idg_rocof}), we can know that the ROCOF ${d f}/{d t}$ has a negative correlation with the equivalent inertia $H_{sys}$ at the same power deficit $\Delta P$.
\begin{equation}
\begin{aligned}
\label{eqn:time_freq_setting}
& \dfrac{d f}{d t} = \dfrac{f_{no}\Delta P}{ 2 H_{sys} }
\end{aligned}
\end{equation}
where $H_{sys}$ can be obtained based on (\ref{eqn:dg_rocof}) and (\ref{eqn:idg_rocof}) according to \cite{Multi2013Gu}. Hence, the time $t_{tr}$ that the frequency drops to the frequency setting is represented by
\begin{equation}
\begin{aligned}
& t_{tr} = \dfrac{(f_{no} - f_{tr})}{{d f}/{d t}}= \dfrac{2 H_{sys} (f_{no} - f_{tr})}{f_{no}\Delta P}
\end{aligned}
\end{equation}

Thus, due to the low inertia $H_{sys}$, the microgrid suffers from larger ROCOF at the same power deficit compared with the conventional power system. When power imbalance happens, the time $t_{fa}$ that the frequency drops to the unsafe frequency $f_{fa}$ in microgrids is less than that in the conventional power system. Based on (\ref{eqn:time_freq_setting}), $t_{fa}$ can also be calculated by
\begin{equation}
\begin{aligned}
& t_{fa} = \dfrac{(f_{no} - f_{fa})}{{\Delta f}/{\Delta t}}= \dfrac{2 H_{sys} (f_{no} - f_{fa})}{f_{no}\Delta P}
\end{aligned}
\end{equation}

The time $t_{rp}^{}$ for ROCOF relay process equals to $t_{fa} - t_{tr}$.
Thus, $t_{rp}^{}$ has a negative correlation with $f_{tr}$. From the Fig. \ref{fig:frequency_regulation}, we can know that the time for load shedding can be calculated as $t_{ls} = t_{rp}^{} - t_{gi} - t_{ad} $. Hence, $t_{rp}^{}$ has a positive correlation with $f_{tr}$.
Considering that $t_{gi}$ is always larger than zero, $t_{ls}$ may be below zero with a large under-frequency setting. In other words, the ROCOF relay may not carry out the load shedding process in time before the microgrid collapses.
Therefore, the large frequency setting is not suitable for the proposed solution in microgrids, especially in islanding mode. Consequently, the small frequency setting $\Delta f = 0.5$Hz is selected in our solution.

The time delay is employed for improving the safety and minimizing the possibility of false operation (nuisance tripping), which usually ranges from 50ms to 500ms \cite{Evaluation2008Ten}. Firstly, the GID operates based on the consensus method, which has a good performance in reducing the measurement noise and oscillation of frequency \cite{zhao2013optimal}. Thus, this method can reduce the effect of nuisance tripping.
Secondly, the smaller time delay has little impact on the avoidance of false operation. The larger time delay leads to a longer detection time of the faults, which may cause that the load shedding process does not respond promptly. In our proposed solution, $t_{gi}$ can be considered as a part of the total time delay $t_d = t_{gi}+t_{ad}$. The smaller $t_{gi}$ gives a wide tuning range for the time delay $t_{ad}$. $t_{gi}$ is determined by communication topology, transition matrix $A$, and communication protocol. The former two can be adjusted according to the physical space and consensus theory \cite{Convergence2009Alex}. These topics are out of the scope of this paper.
The MMST protocol is proposed to reduce the time delay caused by the third one, which is introduced in the following section. To sum up, the total response time $t_{to}$ must be less than $t_{fa}$, which can be described as
\begin{equation}
\begin{aligned}
& t_{to} = t_{tr} + t_{gi} + t_{ad} < t_{fa}
\end{aligned}
\end{equation}
Additionally, the detailed simulation to analyse the performance of the proposed solution at different total time delays are conducted in the subsection \ref{subsec:plugout_line_topo}.
}

\section{ Multicast Metropolis Schedule based on TDMA for Load Shedding }
\label{sec:protocol_design}
The response time of DLSS method is determined by the convergence of the GID process, which is also related to communication protocol. Eq. (\ref{eqn:average_consensus}) can be expressed in the time based update format as follow
\begin{equation}
\label{eqn:uk_update_time}
x_i^{(k+1)t_{one}} = \sum_{j=1}^{N} a_{ij}x_{j}^{ (k t_{one})} ,
\end{equation}
where $t_{one}$ denotes the time period of each iteration. Thus, the convergence time can be represented by $T_{gi} = N_{up} t_{one}$. $ N_{up} $ is the iterations of convergence.
{$t_{one}$ consists of a communication time delay and a calculation time delay.} The calculation time delay can be neglected because the computing performance of each bus agent is powerful enough. Thus $t_{one}$ has a direct relationship with the communication protocol.

Since the TDMA scheme is a collision-free protocol, it is adopted to improve the convergence speed and guarantee the stability of DLSS method.
Thus, the MMST protocol for load shedding is proposed here.
The IEEE 802.11 protocol is employed for analysis without loss of generality.

\subsection{Time Slot Allocation for Multicast Metropolis Schedule}
The slot assignment to improve the channel utilisation is the main problem in this protocol design.
In each update process of the proposed solution, each agent needs to exchange data with all the neighbour agents.
If the unicast mode is adopted, each updating process needs time slots $S=2|\mathcal{E}|$, where $|\mathcal{E}|$ is the number of the transmission link.
The multicast mode is adopted to reduce slots $S$ in each update, i.e., each agent sends information to all his neighbour agents in the same slot. So the used time slots $S$ is reduced from $2|\mathcal{E}|$ to $ N $.
Due to the distributed nature and the sparsity characteristic of microgrids, $S$ can be further reduced by realising concurrent transmission. However, the concurrent transmission may have the hidden terminal problem which causes the packet collision.
Thus, the objective is to obtain the optimal transmission schedule to minimise the used time slot $S$ subject to two constraints. Firstly, each agent has one non-private slot to transmit information; Secondly, the two-hop neighbours $\mathcal{N}_2(i)$ of agent $i$ cannot transmit in the same slot.

The issue is a vertex colouring problem which has been proved to be an NP-hard problem.
Thus we design a heuristic algorithm to obtain the sub-optimal slot assignment inspired by \cite{Brelaz1979New}. This algorithm has two loops to find the suboptimal transmission schedule. The outer loop generates slot and allocates it to a maximum degree agent firstly in this slot until all agents have transmission slot. The inner loop is to find the maximum number of concurrent transmission and the corresponding agent group in this slot.

\begin{algorithm}[htb]
  \caption{Time Slot Allocation for Multicast Metropolis Schedule}
  \label{alg:synchronizing}
  \begin{algorithmic}[1]
  \REQUIRE agent set $\mathcal{N}$; two-hop neighbour set $\mathcal{N}_2(i)$ of agent $i$;
  \ENSURE slot number $s $;
    \STATE $s \leftarrow 0$; $\mathcal{N}^\prime \leftarrow \emptyset$; $\mathcal{N}^* \leftarrow \mathcal{N}$;
    \WHILE { $\mathcal{N}^* \neq \emptyset$ }
      \STATE generate one slot $s \leftarrow s+1 $;
      \STATE $i \leftarrow$ extract a maximal degree agent in $\mathcal{N}^*$;
      \STATE $ \mathcal{N}^{\prime} \leftarrow \mathcal{N}_2(i) +\{i\}$;
      \WHILE { $ \mathcal{N}^{\prime} \neq \mathcal{N} $ }
        \STATE $j \leftarrow$ extract a maximal degree agent in $ \mathcal{N} - \mathcal{N}^{\prime} $;
        \STATE $\mathcal{N}^* \leftarrow \mathcal{N}^* - \{j\}$;
        \STATE $\mathcal{N}^{\prime} \leftarrow \mathcal{N}^{\prime} \cup \mathcal{N}_2(j) \cup \{j\}$;
      \ENDWHILE
    \ENDWHILE
  \end{algorithmic}
\end{algorithm}

\subsection{Frame Design of Multicast Metropolis Schedule}
Due to the slot allocation algorithm in MMST, the packet loss caused by the hidden terminal problem is avoided. However, packet loss caused by link quality cannot be avoided. Thus the data frame is defined in Fig. \ref{fig:frame_structure} for reliable packet delivery.

The \textit{Index\_DATA} includes the consensus indexes of the currently transmitted data, which is used for consensus operation in the GID and utilisation level update.
\textit{Index\_NEIG} and \textit{Bitmap} indicate whether the data of neighbour agents have been received successfully, which are inserted into the data frame.
\textit{Index\_NEIG} contains the neighbour indexes of the agent who transmits the data frame.
The \textit{Status\_NEIG-i} in \textit{Bitmap} is the status of the received data from the $i$-th neighbour agent.
If the previous data frame from the $i$-th neighbour agent has been received successfully, \textit{Status\_NEIG-i} is set to $1$, otherwise set to $0$.
If the previous data frame fails to be received, the neighbour agent will add it in history data part of the data frame and retransmit with the new data.
Thus, retransmission of the lost packet is realised to improve the transmission reliability.
The status data include two parts: current status data and history status data. Current status data are the update data of each agent, and the history status data are the previous data which have not been received correctly. The length of data can be adjusted according to the system requirement. In our method, the data that needs to be updated include load information $ P_{L_i} $, $\rho_{g,i}$, power deficit $\Delta P_i$, and utilization level $u_i$, $\tilde u_i$.

\begin{figure}[ht]
\begin{center}
\vspace{-2ex}
\setlength{\belowcaptionskip}{-0.2cm}
\includegraphics[width = 0.48 \textwidth]{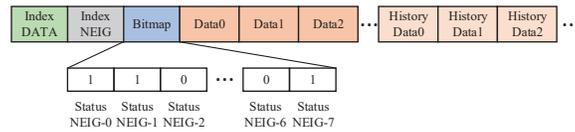}
\caption{ DATA Frame Structure.}
\captionsetup{justification=centering}
\label{fig:frame_structure}
\end{center}
\vspace{-2ex}
\end{figure}

\section{Simulations}
\label{sec:simulation}
In this section, the proposed distributed load shedding solution is tested using NS3-Matlab co-simulator which is implemented based on the co-simulation structure \cite{Zhizhang2016NS3}. The microgrid is modelled in Matlab/Simulink, and the network communication is simulated in NS3.
The two simulators can exchange message by the interactive interface part which is designed based on socket model. {The co-simulation framework is shown in Fig. \ref{fig:configuration(4)}.}

\begin{figure}[htbp]
\centering
\vspace{-2ex}
\subfigure[]{\label{fig:configuration(4)}
\includegraphics[width=0.32 \textwidth]{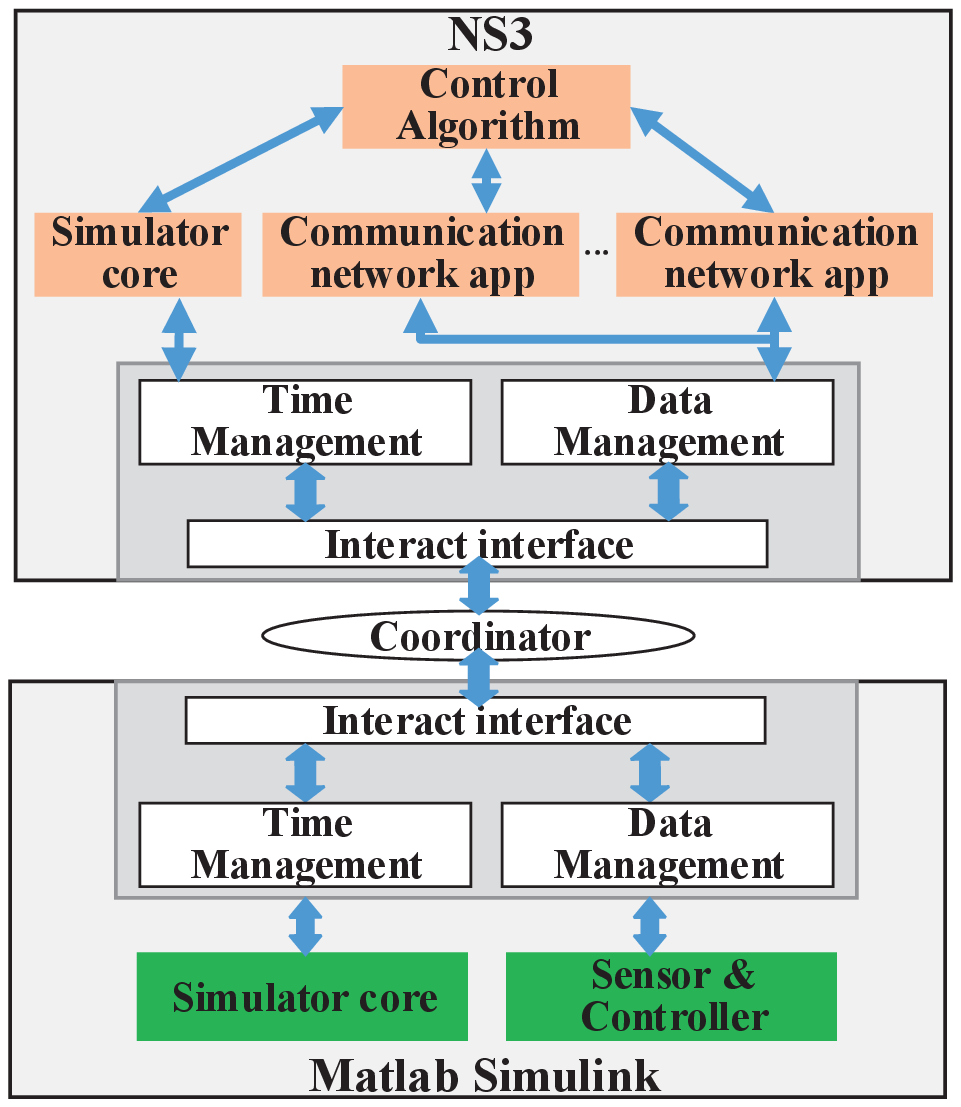}}
\subfigure[]{\label{fig:configuration(3)}
\includegraphics[width=0.48 \textwidth]{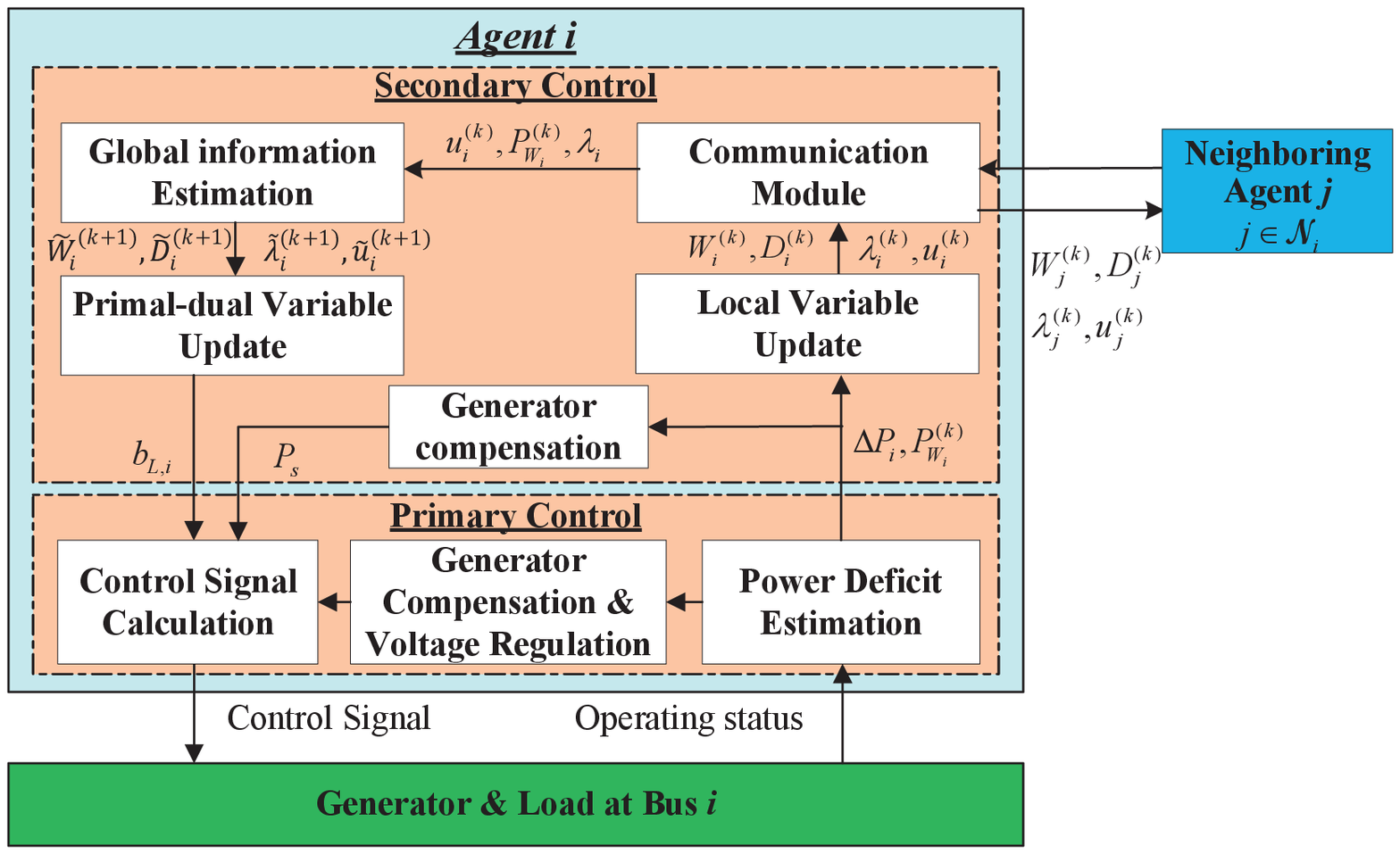}}
\caption{  {(a) Co-simulation framework based on NS3 and Matlab. (b) Management operation diagram. } }
\label{fig:configuration}
\vspace{-2ex}
\end{figure}

The agents communicate via a wireless network which has a same communication topology as the power transmission topology.
The transition matrix $A$ employs
an improved \textit{Metropolis} method \cite{xiao2007distributed}, which is defined as:
\begin{equation}
\label{equ:trans_matrix}
a_{ij}= \left\{
\begin{array}{ll}
\frac {1}{\max\{N_i,N_j\}+1} & \quad  j \in \mathcal{N}(i) \\
1- \sum_{j \in \mathcal{N}_i} \frac{1}{\max\{N_i,N_j\}+1} & \quad  i=j \\
0 & \quad \text{otherwise} ,
\end{array}
\right.
\end{equation}
where $N_i$ denotes the number of neighbour agent $i$, and $\mathcal{N} (i)$ represents the set of agent $i$.

The management operation of each agent is shown in Fig. \ref{fig:configuration(3)}. The hierarchical management strategy consists of two control levels.
The secondary level is responsible for exchanging and updating the utilisation level and setting the reference power of loads. The communication module exchanges local information with its neighbour agents.
The primary level is used for real power tracking while satisfying other constraints including reactive power and voltage regulation.
There are multi-priority loads at each bus. {$\kappa_f$ and $\kappa_v$ are set to $1.0$.} {The trigger frequency and the safety threshold frequency are set to 49.5 Hz and 48 Hz.} The weight factor $w_{g}$ of the three priority loads are set to 1, 2 and 5.
In normal condition, SG is just used for voltage regulation and generates power at a low level. Once fault happens, the SG generates power to compensate the deficit.

\subsection{{Case 1: Islanding in a Microgrid with Radial Topology}}
The 6-bus system with radial topology is illustrated in Fig. \ref{fig:configuration(1)}.
This system contains different types of DGs, such as SG, PV, wind turbine (WT). The information of the generators and loads are shown in Table \ref{tab:ders} and \ref{tab:para_load}.
The ramp-up and ramp-down rates of the SG are both set to 40 kW/s, which determine the maximum compensation rate.
In this case, the communication topology is depicted in Fig. \ref{fig:configuration(2)}.
\begin{figure}[htbp]
\centering
\vspace{-2ex}
\subfigure[]{\label{fig:configuration(1)}
\includegraphics[width=0.48 \textwidth]{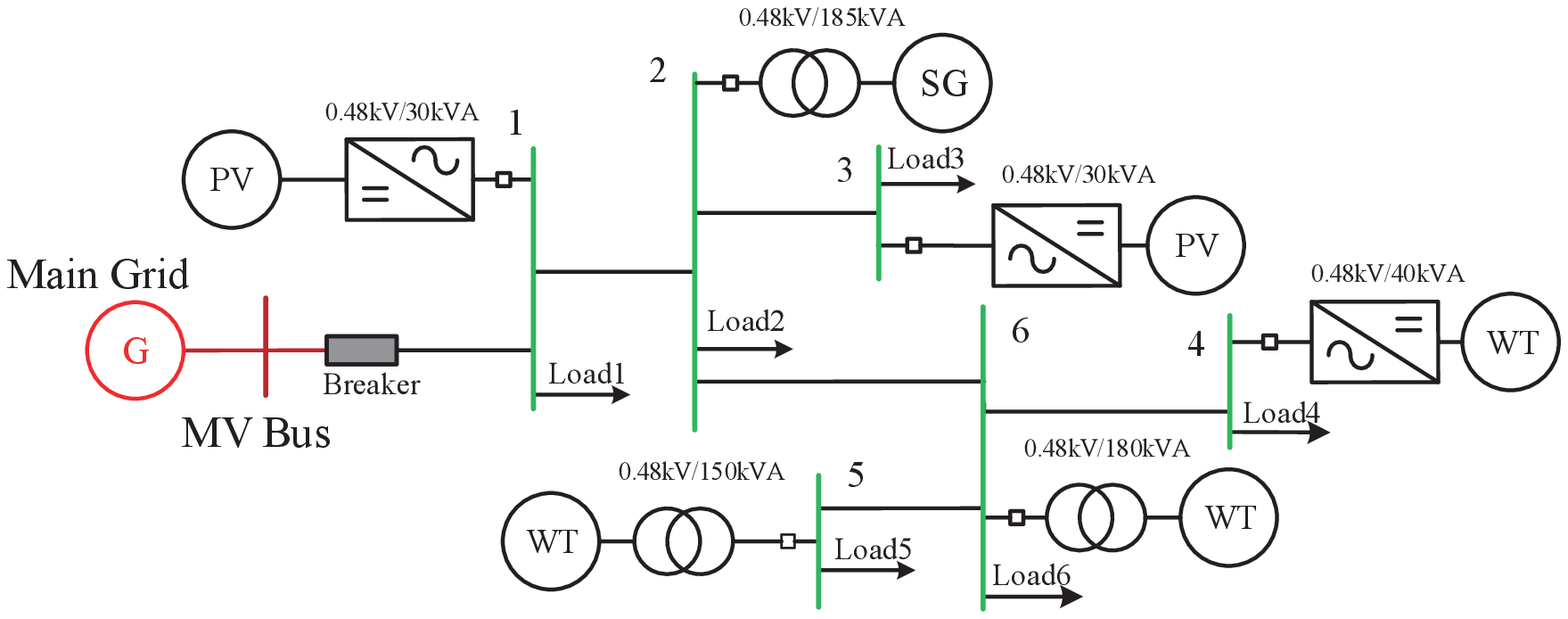}}
\subfigure[]{\label{fig:configuration(2)}
\includegraphics[width=0.3 \textwidth]{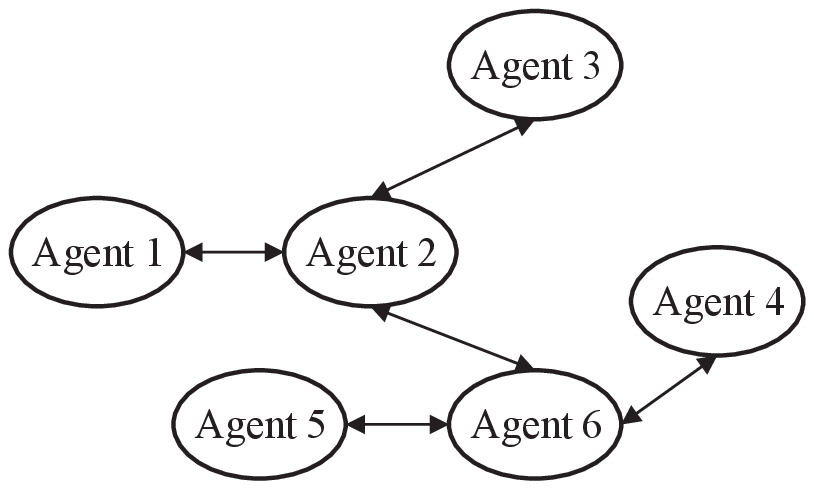}}
\caption{ {(a) 6-bus microgrid with radial topology. (b) Communication topology.} }
\end{figure}

\begin{table}[htbp]
\centering
\caption{Parameters of DERs.}
\label{tab:ders}
\begin{tabular}{cccclc}
\hline
Bus & \begin{tabular}[c]{@{}c@{}}DG\\ Types\end{tabular} & \begin{tabular}[c]{@{}c@{}}Capacity\\ (kVA)\end{tabular} & Types & \begin{tabular}[c]{@{}c@{}}($H_{cg} / \xi$)\end{tabular} & \begin{tabular}[c]{@{}c@{}}Control \\ Mode \end{tabular}  \\ \hline
1   & PV          & 30                        &   Inverter-based    &     $\xi$ = 1.5e-3 & MPPT                \\
2   & SG          & 185                       &   Conventional      &     $H_{cg}$ = 1.68 & PQ-V/f              \\
3   & PV          & 30                        &   Inverter-based    &     $\xi$ = 1.5e-3  & MPPT               \\
4   & WT          & 40                        &   Conventional      &     $H_{cg}$ = 0.68 & PQ              \\
5   & WT          & 150                       &   Conventional      &     $H_{cg}$ = 1.38 & PQ              \\
6   & WT          & 180                       &   Conventional      &     $H_{cg}$ = 1.46 & PQ              \\ \hline
\end{tabular}
\end{table}
\begin{table}[htbp]
\centering
\caption{Parameters of loads.}
\label{tab:para_load}
\begin{tabular}{ccccccc}
\hline
\multirow{2}{*}{Load} & \multirow{2}{*}{\begin{tabular}[c]{@{}c@{}}Real Power\\ (kW)\end{tabular}} & \multirow{2}{*}{\begin{tabular}[c]{@{}c@{}} $P_{L,i}$\\ (kW)\end{tabular}} & \multirow{2}{*}{$N_{L,i}$} & \multicolumn{3}{c}{$\rho_{g,i}$}  \\
                      &  & & &  $\rho_{1,i}$ &  $\rho_{2,i}$  &   $\rho_{3,i}$ \\ \hline
Load1   &  100  & 2 & 50  &   0.5      &  0.3      &  0.2     \\
Load2   &  120  & 2 & 60  &   0.6      &  0.2      &  0.2     \\
Load3   &  150  & 2 & 75  &   0.5      &  0.3      &  0.2     \\
Load4   &  100  & 2 & 50  &   0.3      &  0.5      &  0.2     \\
Load5   &  100  & 2 & 50  &   0.3      &  0.5      &  0.2     \\
Load6   &  120  & 2 & 60  &   0.5      &  0.3      &  0.2     \\ \hline
\end{tabular}
\end{table}

When $t = 2 s$, the distributed microgrid is disconnected from the main grid. The power generation cannot restore system frequency immediately. As a result, the system frequency starts to drop rapidly after this disturbance.

\subsubsection{Global information discovery}
The GID process is always carried out periodically in normal operation mode. When power imbalance occurs, this process is triggered by the drop of system frequency.
The global information that contains the power deficit $\Delta P_i$, the total real power of loads $P_{L}^{\max}$, and the ratio of the $g$-th loads $\rho_g$.
The iteration processes of the GID are shown in Fig. \ref{fig:agent_update}.  $P_{V_1}$ can be obtained from Fig. \ref{fig:agent_update}, so $\rho_1$ is calculated by $ (P_{V_1} / P_{L}^{\max}) $.
The power deficit of different DGs $\Delta P_i$ are estimated by different methods in (\ref{eqn:dg_rocof}) and (\ref{eqn:idg_rocof}).

\begin{figure}[htbp]
\centering
\vspace{-2ex}
\subfigure[]{\label{fig:agent_update}
\includegraphics[width=0.48 \textwidth]{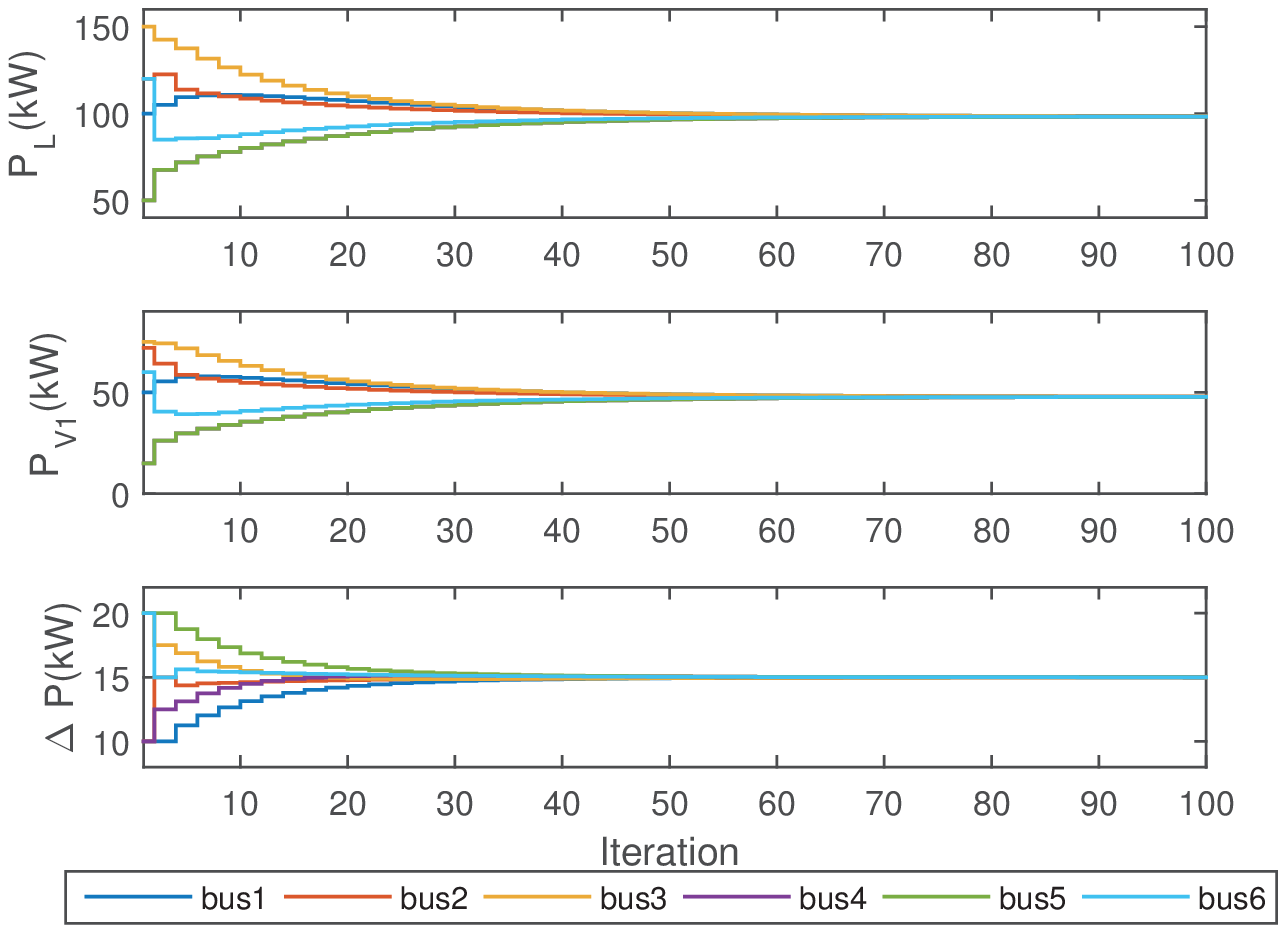}}
\subfigure[]{\label{fig:coordination_error}
\includegraphics[width=0.48 \textwidth]{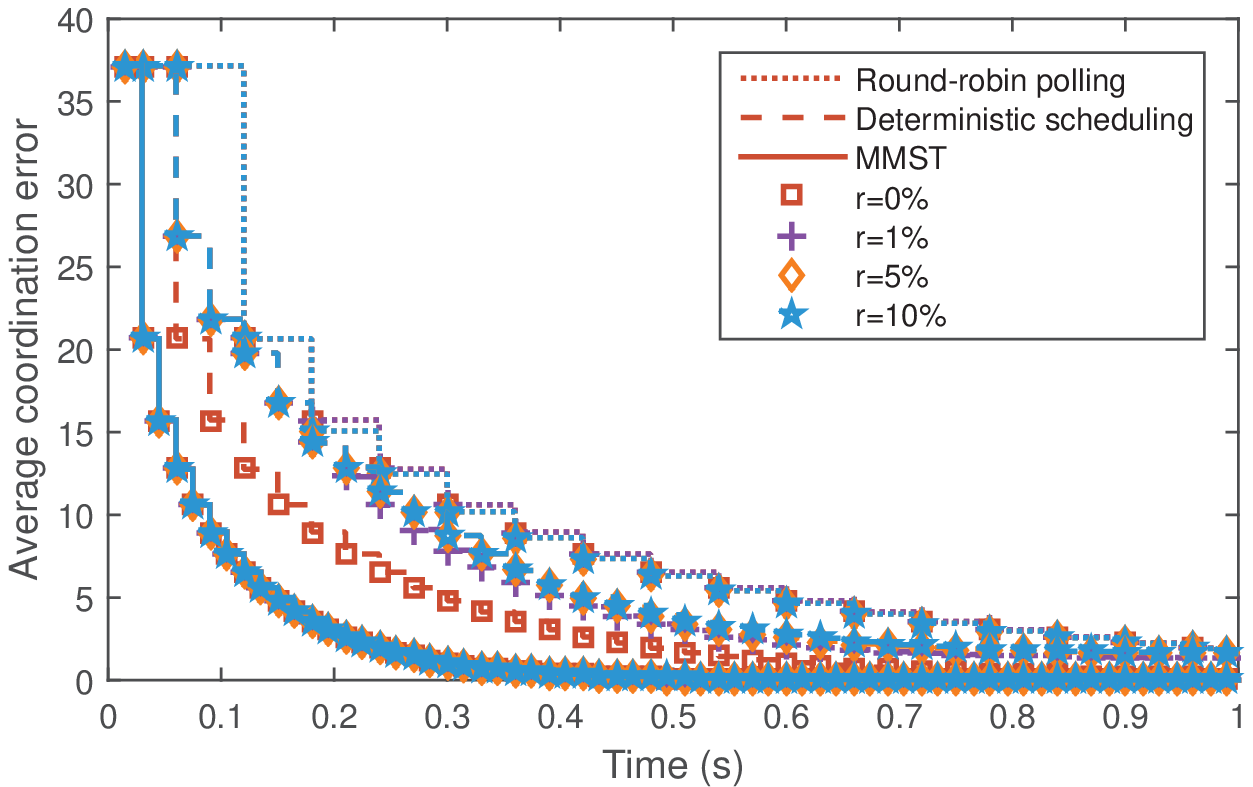}}
\caption{ (a) Global information discovery. (b) Coordination error comparison.}
\end{figure}

The process of power deficit discovery is used for the convergence analysis of the proposed MMST protocol, round-robin polling mechanism, and deterministic scheduling \cite{liang2012multiagent}).
In this case, the time slot of communication protocol is set to $5$ ms, which can meet the per-hop latency in sub-6 GHz wireless technology.
The three protocols need to allocate 4, 6, and 10 time-slots for one iteration $t_{one}$, respectively. Four conditions with different packet loss rate $r$ are considered for analysis. The results are shown in Fig. \ref{fig:coordination_error} and Table \ref{tab:convergence_global}.

\begin{table}[htbp]
\centering
\caption{Performance Comparison in Global Information Discovery. }
\label{tab:convergence_global}
\begin{tabular}{ccccccc}
\hline
{\begin{tabular}[c]{@{}c@{}}Packet \\ Loss Rate  \end{tabular}} & \multicolumn{2}{c}{\begin{tabular}[c]{@{}c@{}} Round-robin \\ polling \end{tabular}}  & \multicolumn{2}{c}{\begin{tabular}[c]{@{}c@{}}Deterministic\\ scheduling \end{tabular}} & \multicolumn{2}{c}{MMST} \\
 \multicolumn{1}{c}{\multirow{1}{*}{$ r $ }}  & \multicolumn{1}{l}{\multirow{1}{*}{$T_{gi}$ (s)}} & \multirow{1}{*}{$e$} & \multicolumn{1}{l}{\multirow{1}{*}{$T_{gi}$ (s)}}        & \multirow{1}{*}{$e$}        & \multicolumn{1}{l}{\multirow{1}{*}{$T_{gi}$ (s)}}     & \multirow{1}{*}{$e$}     \\ \hline
0\%  &   1.25    & 0\%    &  0.66    &    0\%   &   0.32   &   0\%     \\
1\%  &   1.25    & 1.3\%  &  0.75    &  1.29\%  &   0.32   &   0\%     \\
5\%  &   1.35    & 1.6\%  &  0.81    &  1.61\%  &   0.32   &   0\%     \\
10\% &   1.40    & 1.55\% &  0.84    &  1.54\%  &   0.32   &   0.01\%  \\ \hline
\end{tabular}
\end{table}

The results demonstrate that the convergence time $ T_{gi}$ increases and the relative error $e$ becomes larger with the increase of packet loss rate. The average coordination error $ e_{avg} $ is defined as
\begin{equation}
\label{eqn:average_error}
e_{avg} = \sqrt{\dfrac{1}{N} \sum_{i \in \mathcal{N}} (x_i[t] - x^{*})^2} ,
\end{equation}
where $x^{*}$ is the true value which can be obtained by (\ref{eqn:average_consensus}).
Packet loss has less impact on the performance of MMST and the convergence time of MMST is shorter than that of the others.
The convergence time of the GID is important for load shedding which is a time-sensitive process. The less the convergence time is, the more time the distributed load shedding has. For instance, the convergence of round-robin polling is more than the time that frequency dropped to 48Hz in Fig. \ref{fig:frequency_response_system}. Consequently, the DLSS algorithm does not have enough time to respond to the overload. Therefore, MMST is more suitable for load shedding method.

\subsubsection{ Load shedding process }
In this simulation, we the time step of load shedding $\Delta t = 80$, $t_{ad}=0$ms and $\Delta \tilde P = 120$kW.
The power imbalance occurs at $t = 2 s$, the unused capacity of all the DGs cannot immediately eliminate the power deficiency.
DLSS method is carried out directly based on the information estimated by the global information discovery. Meanwhile, SG generates real power to compensate the deficiency.
Once all the agents obtain the global information, the DLSS method is carried out to shed loads.
Depending on the DLSS method, cooperative load shedding process can be achieved.
The utilisation levels and objective value at each bus are calculated locally.
The utilisation levels are asymptotically converged, and the objective value converges to 0 with $ \tau_k= 1/[ 10(\check P_{L_i}^{2})+ C_\lambda)] $.

\begin{figure}[htbp]
\setlength{\belowcaptionskip}{-0.3cm}
\centering
\subfigure[]{\label{fig:utilization_level_objective_value}
\includegraphics[width = 0.48 \textwidth]{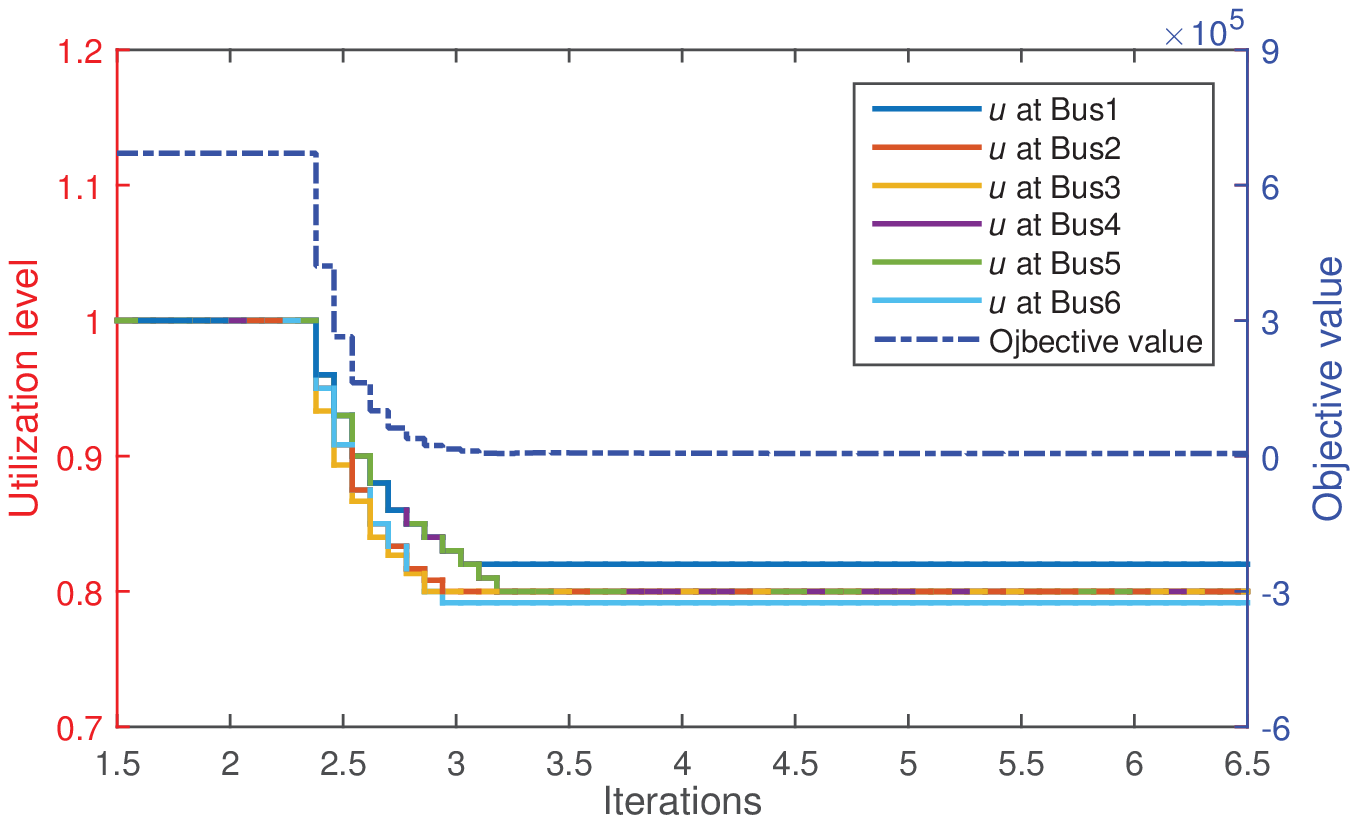}}
\subfigure[]{\label{fig:frequency_response_system}
\includegraphics[width = 0.48 \textwidth]{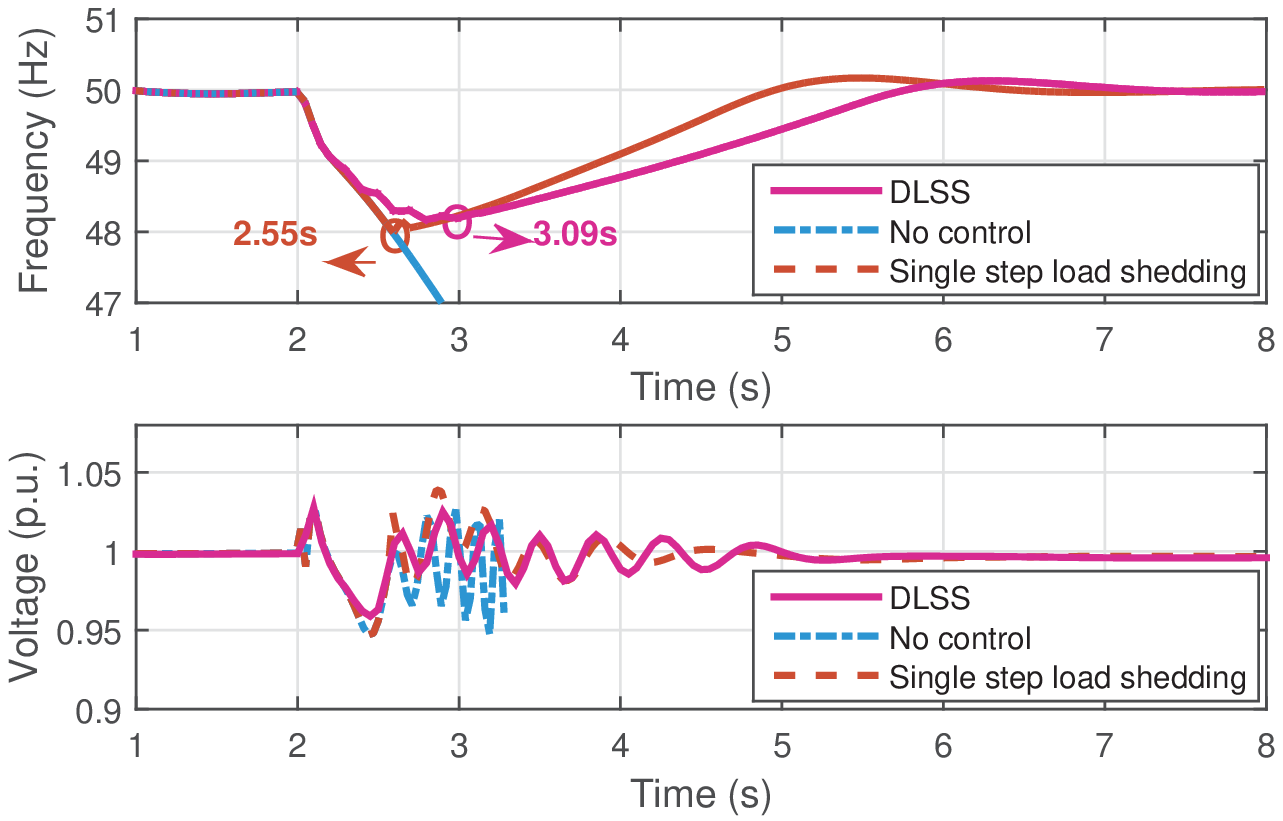}}
\subfigure[]{\label{fig:load_shedding_amout_1}
\includegraphics[width = 0.48 \textwidth]{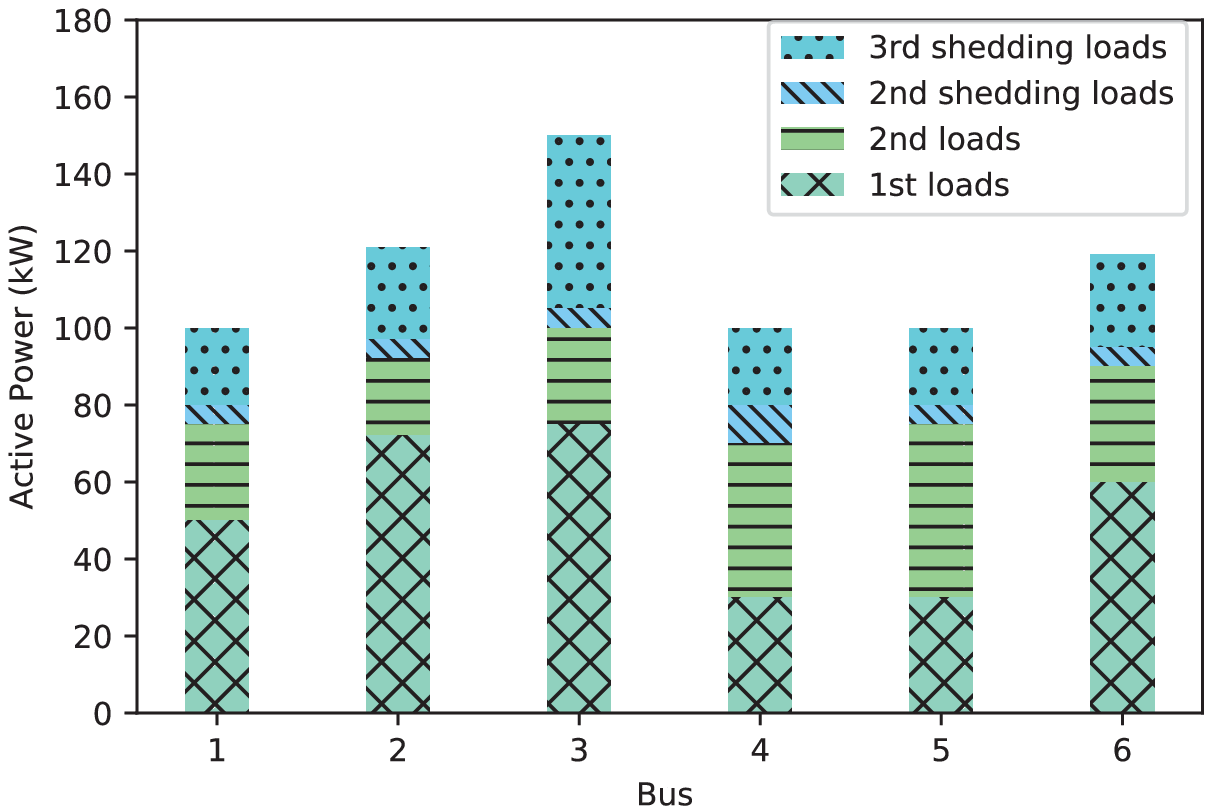}}
\caption{ (a) Global information discovery. (b) Coordination error comparison. (c) Load shedding at each bus. }
\label{fig:results}
\end{figure}

From Fig. \ref{fig:frequency_response_system}, we can observe that the system frequency drops close to 48.00 Hz, then gradually recovers to the rated value 50 Hz.
Due to the gradually shedding load, the ROCOF $df/dt$ is reduced.
In the process of load shedding, the voltage fluctuates.
The voltage response of SG in Fig. \ref{fig:frequency_response_system} shows that the proposed DLSS method will not cause under voltage during the process of the whole control. The frequency will drop below 47.5 Hz and be unsafe with the round-robin polling and deterministic schedule in this scenario. Because their convergence time in table \ref{tab:convergence_global} is larger than the drop time to be unsafe.
Additionally, the final load-shedding amount at each bus is depicted in Fig. \ref{fig:load_shedding_amout_1}.
The bars filled with dots and hatched lines shown in Fig. \ref{fig:load_shedding_amout_1} indicate the load-shedding amount.
We can observe that the first-grade loads remain unchanged, part of the second-grade loads and all the third loads are disconnected.
In this case, it is evident that DLSS method can implement stable load shedding.
Simulation results demonstrate the effectiveness of the proposed scheme to maintain frequency stability during a large disturbance.

\subsubsection{Load-shedding amount analysis}
The compensation power amount of distributed SG is related to the adjusted time of DLSS method, which is impacted by step size $ \tau_k $.
The more time the SG has for compensation, the less of loads should be shed. The simulation is carried out with different parameters $\tau_k$ and power deficiency $\Delta \tilde P$.
The results are shown in Table. \ref{tab:shedding_amount_comparsion}.
The five parameters of $\tau_k$ are $ 1/[ 20\check P_{L_i}^{2}(w_{G}^2+ C_\lambda)] $, $ 1/[ 15\check P_{L_i}^{2}(w_{G}^2+ C_\lambda)] $, $ 1/[ 12\check P_{L_i}^{2}(w_{G}^2+ C_\lambda)] $, $ 1/[ 10\check P_{L_i}^{2}(w_{G}^2+ C_\lambda)] $, $ 1/[ 4 \check P_{L_i}^{2}(w_{G}^2+ C_\lambda)] $.
We can observe that the number of shedding steps is reduced and the load-shedding amount at each step is increased with the increase of $\tau_k$.
Larger shedding amount can realise the supply-demand balance faster, but there is not sufficient time for generator compensation. Thus, the load-shedding amount isn't reduced significantly.
Because the smaller parameters $\tau_k$ have little impact on the reduction of frequency derivative, it cannot avoid the frequency dropping to the unsafe range.
From table \ref{tab:shedding_amount_comparsion}, it can be seen that the system frequency drops to the unsafe range when the power deficits are 160 and 200 kW with the smallest $\tau_k$. In these cases, the safety threshold shedding is executed. Thus, the final power deficit $\Delta {\tilde P}$ by load shedding is more than the case with larger $\tau_k$.
With the same $ \tau_k$, the load-shedding amount at each step and the number of shedding steps are affected by the power deficit $\Delta \tilde P$.
A rapid frequency decline can be decreased after executing the steps with a large load-shedding amount. Consequently, the frequency would not drop too fast to reach the unsafe range.
When the power deficit is relatively small, the initial steps have small load-shedding amount to avoid shedding loads too quickly. Hence, the generator has enough time to compensate the power deficit.

\begin{table}[htbp]
\centering
\caption{Load-shedding amount with different step sizes.}
\label{tab:shedding_amount_comparsion}
\begin{tabular}{ccllll}
\hline
\multirow{2}{*}{\begin{tabular}[c]{@{}c@{}}Power deficit \\ $\Delta \tilde P$ (kW) \end{tabular}} & \multicolumn{5}{c}{ $\tau_k$ } \\
               & \multicolumn{1}{l}{\multirow{1}{*}{1}} & \multirow{1}{*}{2} & \multirow{1}{*}{3}   & \multirow{1}{*}{4} & \multirow{1}{*}{5} \\ \hline
80             &  58.0   & 72.0   & 72.0   &  72.0  & 74.0              \\ \hline
120            &  74.0   & 92.0   & 92.0   &  94.0  & 94.0             \\ \hline
160            &  156.0      & 150.0  & 152.0  &  152.0 & 154.0             \\ \hline
200            &  198.0      & 190.0  & 192.0  &  192.0 & 194.0              \\ \hline
\end{tabular}
\end{table}

\subsection{Case 2: Disconnection in a Microgrid with Line Topology}
\label{subsec:plugout_line_topo}
The 6-bus system with line topology is illustrated in Fig. \ref{fig:case2_configuration(1)}.
The information of the generators and loads are shown in Table \ref{tab:case2_ders} and \ref{tab:case2_para_load}.
The ramp-up and ramp-down rates of the SG are also both set to 40 kW/s, which determine the maximum compensation rate.
The communication topology of this system is depicted in Fig. \ref{fig:case2_configuration(2)}.
A disconnection of the WT at bus 5 is simulated in the islanding mode.

\begin{figure}[htbp]
\centering
\vspace{-20pt}
\subfigure[]{\label{fig:case2_configuration(1)}
\includegraphics[width=0.48 \textwidth]{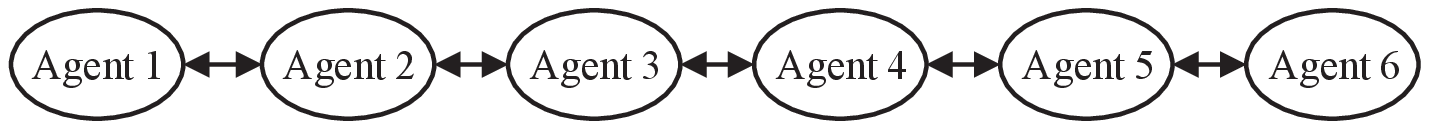}}
\subfigure[]{\label{fig:case2_configuration(2)}
\includegraphics[width=0.45 \textwidth]{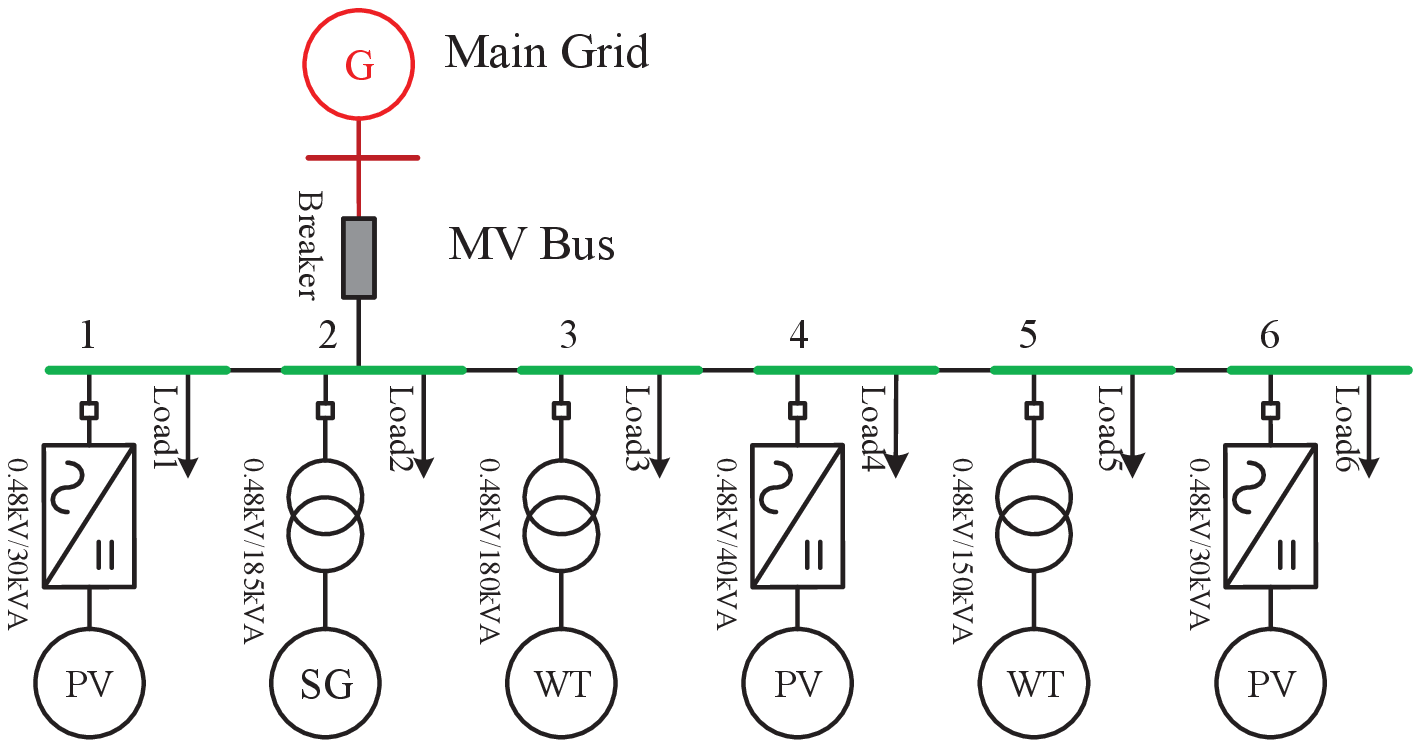}}
\caption{ {(a) 6-bus microgrid with line topology. (b) Communication topology.} }
\end{figure}

\begin{table}[htbp]
\centering
\caption{{Parameters of DERs.}}
\label{tab:case2_ders}
\begin{tabular}{cccclc}
\hline
Bus & \begin{tabular}[c]{@{}c@{}}DG\\ Types\end{tabular} & \begin{tabular}[c]{@{}c@{}}Capacity\\ (kVA)\end{tabular} & Types & \begin{tabular}[c]{@{}c@{}}($H_{cg} / \xi$)\end{tabular} & \begin{tabular}[c]{@{}c@{}}Control \\ Mode \end{tabular}  \\ \hline
1   & PV  & 30    &   Inverter-based  &  $\xi$ = 1.5e-3  & MPPT     \\
2   & SG  & 300   &   Conventional    &  $H_{cg}$ = 2.62.6
 & PQ-V/f   \\
3   & WT  & 150   &   Conventional    &  $H_{cg}$ = 1.38 & PQ     \\
4   & PV  & 40    &   Inverter-based  &  $\xi$ = 1.5e-3  & MPPT     \\
5   & WT  & 150   &   Conventional    &  $H_{cg}$ = 1.38 & PQ       \\
6   & PV  & 30    &   Inverter-based  &  $\xi$ = 1.5e-3  & MPPT     \\ \hline
\end{tabular}
\end{table}
\begin{table}[htbp]
\centering
\caption{{Parameters of loads.}}
\label{tab:case2_para_load}
\begin{tabular}{ccccccc}
\hline
\multirow{2}{*}{Load} & \multirow{2}{*}{\begin{tabular}[c]{@{}c@{}}Real Power\\ (kW)\end{tabular}} & \multirow{2}{*}{\begin{tabular}[c]{@{}c@{}} $P_{L,i}$\\ (kW)\end{tabular}} & \multirow{2}{*}{$N_{L,i}$} & \multicolumn{3}{c}{$\rho_{g,i}$}  \\
                      &  & & &  $\rho_{1,i}$ &  $\rho_{2,i}$  &   $\rho_{3,i}$ \\ \hline
Load1   &  80   & 2 & 40  &   0.5      &  0.4      &  0.1     \\
Load2   &  120  & 2 & 60  &   0.6      &  0.2      &  0.2     \\
Load3   &  140  & 2 & 70  &   0.5      &  0.3      &  0.2     \\
Load4   &  80   & 2 & 40  &   0.4      &  0.5      &  0.1     \\
Load5   &  140  & 2 & 70  &   0.3      &  0.5      &  0.2     \\
Load6   &  100  & 2 & 50  &   0.5      &  0.3      &  0.2     \\ \hline
\end{tabular}
\end{table}

\begin{table}[htbp]
\centering
\caption{{Performance Comparison in Global Information Discovery. }}
\label{tab:case2_convergence_global}
\begin{tabular}{ccccccc}
\hline
{\begin{tabular}[c]{@{}c@{}}Packet \\ Loss Rate  \end{tabular}} & \multicolumn{2}{c}{\begin{tabular}[c]{@{}c@{}} Round-robin \\ polling \end{tabular}}  & \multicolumn{2}{c}{\begin{tabular}[c]{@{}c@{}}Deterministic\\ scheduling \end{tabular}} & \multicolumn{2}{c}{MMST} \\
 \multicolumn{1}{c}{\multirow{1}{*}{$ r $ }}  & \multicolumn{1}{l}{\multirow{1}{*}{$T_{gi}$ (s)}} & \multirow{1}{*}{$e$} & \multicolumn{1}{l}{\multirow{1}{*}{$T_{gi}$ (s)}}        & \multirow{1}{*}{$e$}        & \multicolumn{1}{l}{\multirow{1}{*}{$T_{gi}$ (s)}}     & \multirow{1}{*}{$e$}     \\ \hline
0\%  &   0.95    & 0\%    &  0.57    &    0\%   &   0.29   &   0\%     \\
1\%  &   0.95    & 1.26\% &  0.65    &  1.25\%  &   0.29   &   0\%     \\
5\%  &   1.05    & 1.56\% &  0.70    &  1.51\%  &   0.29   &   0\%     \\
10\% &   1.10    & 1.53\% &  0.75    &  1.49\%  &   0.29   &   0.01\%  \\ \hline
\end{tabular}
\end{table}

{
\subsubsection{Global information discovery}
The frequency starts to drop rapidly when the disconnection of the WT in bus 5 takes place at $t = 2 $s. The unused capacity of SG cannot immediately eliminate the power deficiency. If the frequency drops to the trigger frequency $f_{tr}$, the GID is carried out immediately.
In this case, the process of total load power is used for the convergence analysis of the proposed MMST protocol and the other two protocols. The time slot of communication protocol is also set to $5$ ms.
The three protocols need to allocate 3, 6, and 10 time-slots for one iteration $t_{one}$, respectively.
{Without consideration of packet loss, the convergence speed of our proposed protocol is two times faster than the deterministic scheduling}. And it is at least three times faster than the round-robin polling. Four conditions with different packet loss rate $r$ are also considered for analysis. The simulation results are shown in Table \ref{tab:case2_convergence_global}.

From the results, we can observe that the $t_{gi}$ has a direct relationship with the used time slots in each iteration $t_{one}$.
Additionally, packet loss has less impact on the performance of MMST and the convergence time of MMST is shorter than of the others.
Those results verified the effectiveness of the proposed MMST protocol.
}

{
\subsubsection{ Load shedding process }

\begin{figure}[htbp]
\setlength{\belowcaptionskip}{-0.5cm}
\vspace{-20pt}
\centering
\subfigure[]{\label{fig:case2_utilization_level_objective_value}
\includegraphics[width = 0.48 \textwidth]{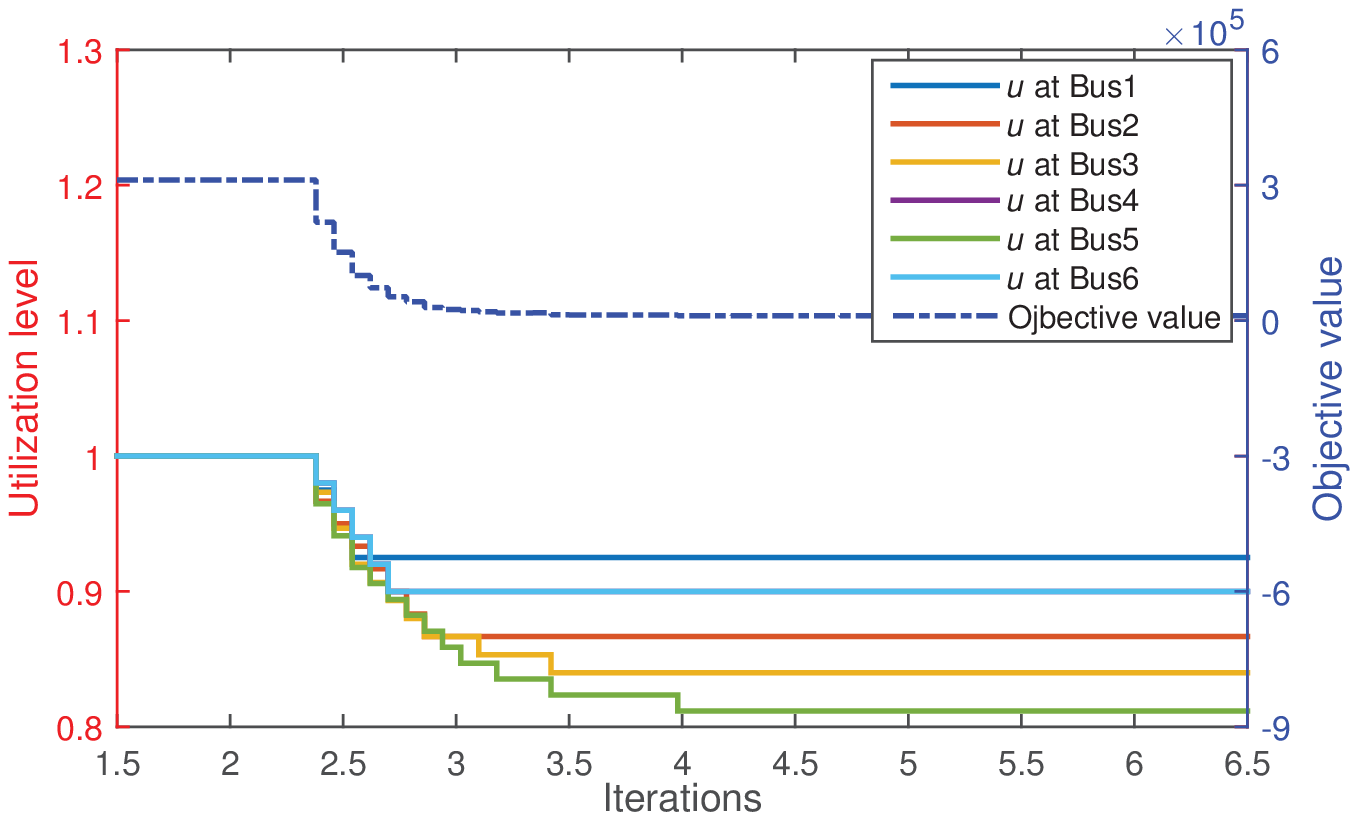}}
\subfigure[]{\label{fig:case2_frequency_response}
\includegraphics[width = 0.48 \textwidth]{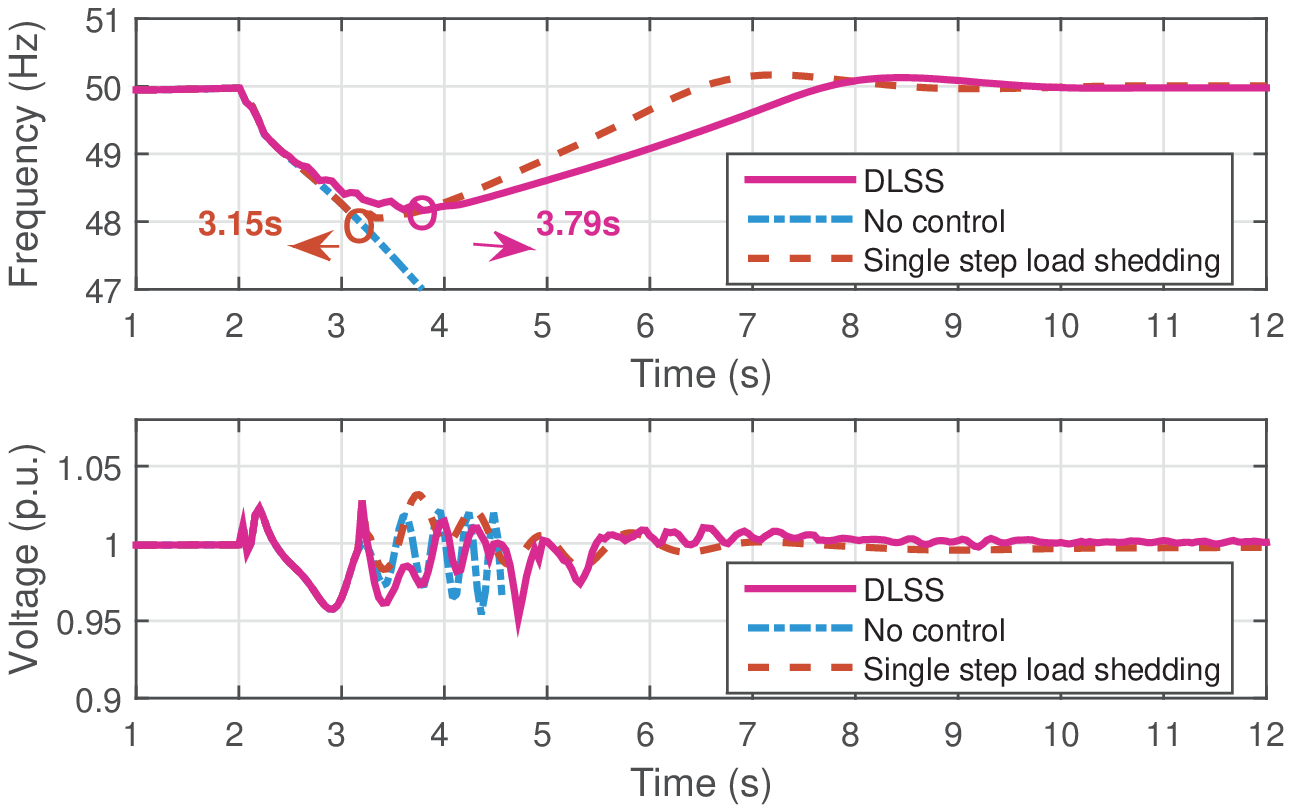}}
\subfigure[]{\label{fig:case2_load_shedding_amout_2}
\includegraphics[width = 0.48 \textwidth]{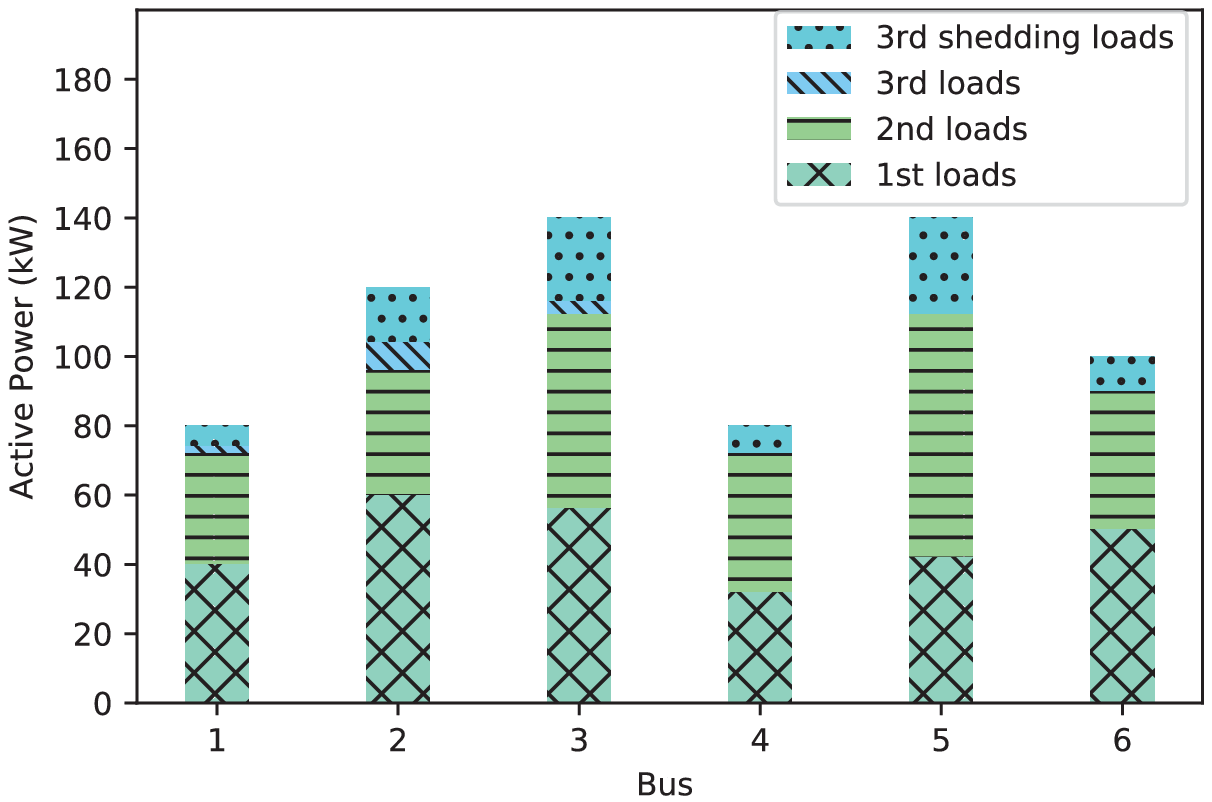}}
\caption{{ (a) Global information discovery. (b) Coordination error comparison. (c) Load shedding at each bus. }}
\label{fig:case2_results}
\end{figure}

When the GID is converged, DLSS method is carried out after the time delay $t_{ad}$ based on the estimated global information. In this simulation, we set $\Delta t$=80 ms, $t_{ad}=0$ms and $\Delta \tilde P = 120$kW. Meanwhile, SG generates real power to compensate the deficiency.
The utilisation levels of all the agents are asymptotically converged, and the objective value converges to 0, which demonstrates that our solution retrieves the power balance.

We can observe from Fig. \ref{fig:case2_frequency_response} that, the system frequency drops close to 48.00 Hz, and then gradually recovers to the rated value 50 Hz. The initial estimated ROCOF equals to 1.75Hz/s. Without load shedding, the time that the frequency drops to 47Hz is approximately 1.8s.
The single step load shedding method executed the operation when the frequency dropped to 48Hz.
Due to the gradual load shedding, the ROCOF $df/dt$ is reduced. Thus, there is more time for power compensation by our solution than other two solutions.
In the process of load shedding by the proposed method, the voltage fluctuates.
The voltage response of SG in Fig. \ref{fig:case2_frequency_response} shows that the proposed DLSS method will not cause under voltage during the process of the whole control.
Additionally, the bars filled with dots and hatched lines shown in Fig. \ref{fig:case2_load_shedding_amout_2} indicate the final load-shedding amount.
We can observe that the first-grade and second-grade loads remain unchanged, part of the third-grade loads are disconnected.
In this case, it is evident that the DLSS method can implement stable load shedding.
Simulation results demonstrate the effectiveness of the proposed DLSS scheme to maintain frequency stability during a large disturbance.
}

{
\subsubsection{Time delay analysis}
Based on the analysis in subsection \ref{subsec:para_setting}, it is known that $t_{gi}$ can be considered as a part of the time delay in the ROCOF relay. Due to $t_{gi}$ cannot be directly adjusted, we test the solution performance at different time delays $t_{ad}$.
The total time delay is calculated by $t_d = t_{gi} + t_{ad}$. The simulation results is shown in Table \ref{tab:case2_shedding_amount_time_delay}. Four different power deficit is considered, which is conducted with different power generations of the WT in bus 5. The step size $\tau_k$ is set to $ 1/[ 15\check P_{L_i}^{2}(w_{G}^2+ C_\lambda)] $.

\begin{table}[thbp]
\centering
\caption{{Load-shedding amount with different time delays.}}
\label{tab:case2_shedding_amount_time_delay}
\begin{tabular}{ccllll}
\hline
\multirow{2}{*}{\begin{tabular}[c]{@{}c@{}}Power deficit \\ $\Delta \tilde P$ (kW) \end{tabular}} & \multicolumn{5}{c}{ $t_{ad}$ (ms) } \\
               & \multicolumn{1}{l}{\multirow{1}{*}{100}} & \multirow{1}{*}{200} & \multirow{1}{*}{400}   & \multirow{1}{*}{600} & \multirow{1}{*}{800} \\ \hline
90             &  40.0   & 42.0   & 52.0   &  60.0  & 70.0     \\ \hline
110            &  70.0   & 76.0   & 84.0   &  94.0  & 102.0    \\ \hline
130            &  102.0  & 106.0  & 116.0  &  124.0 & 130.0    \\ \hline
150            &  128.0  & 134.0  & 142.0  &  150.0 & 150.0    \\ \hline
\end{tabular}
\end{table}

From the results, it can be known that our solution with the smaller delay gets more power compensation. The larger time delay causes a longer response time for the disconnection operation. Consequently, the frequency drops over a longer time. In the worst case, the load shedding process in (\ref{eqn:safety_threshole_shedding}) operates before the DLSS method executes. Thus, the load-shedding amount is similar to the initial estimated power deficit. For example, when $\Delta \tilde P = 150$kW, the load-shedding amounts equal to the power deficit at the time delay $600$ms and $800$ms. In this scenario, the proposed solution has the same performance compared with the conventional single step load shedding scheme. In other words, the power compensation cannot be realised. Consequently, the consumer' experience cannot be improved by reducing the load-shedding amount. Therefore, if $t_{gi}$ is within a reasonable range, we do not need to add another time delay $t_{ad}$.}

\section{Conclusion}
\label{sec:conclusion}
In this paper, a distributed load shedding solution is proposed to shed loads gradually considering the participation of smart homes/buildings. First, the DLSS method is proposed to alleviate the rate of frequency drop. Consequently, the time of frequency to be unsafe is prolonged. Thus, the generators have more time to compensate power deficiency for reducing the load-shedding amount.
Second, an MMST protocol is developed to reduce response time and enhance the reliability of the DLSS method. The simulation results demonstrate that the proposed load shedding solution can maintain the stability of the system frequency and reduce the load-shedding amount.
The future work will concern this issue in the microgrid integrated with energy storage system.

\section*{Acknowledgment}
This work was supported by National Key Research and Development Program of China (2016YFB090190),
National Natural Science Foundation of China (61573245, 61174127, 61521063, and 61633017). This work was also partially supported by Shanghai Rising-Star Program under Grant 15QA1402300 and Shanghai Municipal Commission of Economy and Informatization under SH-CXY-2016-003.

The authors would like to thank the anonymous reviewers for their professional and valuable comments, which have led to the improved version.

\appendix
\label{sec:appdendix}
%



The major steps for proving Theorem 1 is presented here. One key theorem and lemma that used in proof are presented first. The first theorem is Saddle-Point Theorem  \cite{boyd2004convex}.

\begin{theorem}
\label{theorem:theorem_saddle}
The point $(\bm u^{*}, \lambda^{*})$ is primal-dual solution pair of problem (\ref{eqn:lagrange_function}) if and only if there holds
\begin{equation}
\begin{aligned}
\mathcal{L} (\bm u^{*}, \lambda) \leqslant \mathcal{L} (\bm u^{*}, \lambda^{*}) \leqslant \mathcal{L} (\bm u, \lambda^{*}).
\end{aligned}
\end{equation}
\end{theorem}

The lemma 11 of chapter 2.2 in \cite{polyak1987introduction} is used, which is described as follow.

\begin{lemma}
\label{lemma:lemma_positive}
If $b_k$, $d_k$ and $e_k$ are non-negative sequences and satisfy the condition as follow:
\begin{equation}
\begin{aligned}
& \sum_{k=1}^{\infty} c_k< \infty \\
& b_k < b_{k-1} - e_{k-1} + c_{k-1},
\end{aligned}
\end{equation}
the sequence ${b_k}$ converges and $\sum_{k=1}^{\infty} e_k < \infty$.
\end{lemma}


The transition matrix $A$ satisfies that $\sum_{j=1}^{N} a_{ij} = 1$ for all $i,k$ and $\sum_{i=1}^{N} a_{ij} = 1$ for all $j$, which is a doubly stochastic matrix. There exists a scalar $0 < \gamma < 1$ such that $a_{ii} > \gamma$ for all $i$ and $a_{ij} > \gamma$ if $a_{ij} > 0$. The communication graph can be set as similar as power network, so it is strongly connected.

In each iteration of global variable estimation, since $\tau_k$ is positive and non-increasing sequence, it holds that
\begin{equation}
\begin{aligned}
\label{eqn:aulxilary_variable_convergence}
& \sum_{k=1}^{\infty} \tau_k \Big \Vert \tilde W_i^{(k)} - \hat {W}^{(k-1)} \Big \Vert < \infty,
\lim_{k \rightarrow \infty} \Big \Vert \tilde W_i^{(k)} - \hat {W}^{(k-1)}  \Big \Vert =0, \\
& \sum_{k=1}^{\infty} \tau_k \Big \Vert \tilde D_i^{(k)} - \hat {D}^{(k-1)} \Big \Vert < \infty,
\lim_{k \rightarrow \infty} \Big \Vert \tilde D_i^{(k)} - \hat {D}^{(k-1)}  \Big \Vert =0, \\
& \sum_{k=1}^{\infty} \tau_k \Big \Vert \tilde \lambda_i^{(k)} - \hat {\lambda}^{(k-1)} \Big \Vert < \infty,
\lim_{k \rightarrow \infty} \Big \Vert \tilde \lambda_i^{(k)} - \hat {\lambda}^{(k-1)}  \Big \Vert =0,
\end{aligned}
\end{equation}
where
\begin{equation}
\label{eqn:central_variable}
\begin{aligned}
& \hat {W}^{(k)} = \dfrac{1}{N}\sum_{j=1}^{N} {P_{W_j}}^{(k)},
 \hat {D}^{(k)} = \dfrac{1}{N}\sum_{j=1}^{N} \Delta{P_{j}}^{(k)}, \\
& \hat {\lambda}^{(k)} = \dfrac{1}{N}\sum_{j=1}^{N} {{\lambda_j}}^{(k)}.
\end{aligned}
\end{equation}

The local $u_i$ in (\ref{eqn:u_varible_update}) will achieve consensus on the value of $\hat u_i^{(k)}$ asymptotically. We define $\hat u_i^{(k)}$ as
\begin{equation}
\begin{aligned}
\hat u_i^{(k)} & = \left( {u_i}^{(k-1)} -  \tau_k \mathcal{L}_{u_i} \left( \bm{u}^{(k-1)}, \hat {\lambda}_i^{(k-1)} \right) \right)^+ \\
& = \left( {u_i}^{(k-1)} -  2 \tau_k P_{L_{i}}^{\max} \left( \hat {\lambda}_i^{(k-1)} N \tilde D_i^{(k)} - w_{{m+1}} \left( P_{W_t} - P_{W_{\Delta}} - N \tilde W_i^{(k)} \right)  \right) \right)^+,
\end{aligned}
\end{equation}

Based on (\ref{eqn:J_condition2}) and (\ref{eqn:F_condition2}), we can know that
\begin{equation}
\label{eqn:alpha_convergence}
\begin{aligned}
& \sum_{i=1}^{N} \Big \Vert  {{u}}_{i}^{(k)} - {{u}}_{i}^{} \Big \Vert ^2 \\
= & \sum_{i=1}^{N}  \Big \Vert \left( u_i^{(k-1)} - 2 \tau_k P_{L_{i}}^{\max} \left( \tilde {\lambda}_i^{(k)} N \tilde D_i^{(k)} - w_{{m+1}} \left( P_{W_t} - P_{W_{\Delta}} - N \tilde W_i^{(k)} \right) \right) \right)^{+} - u_i^{}  \Big \Vert ^2 \\
\leqslant & \sum_{i=1}^{N}  \Big \Vert u_i^{(k-1)} - u_i^{}  \Big \Vert ^2 + 4 \tau_k^2 N^2 \check P_{L_i} \left( \Delta \tilde P + C_{\lambda} w_{G} P_{W_\Delta} \right)^2  \\
& - \sum_{i=1}^{N}  2 \tau_k  \check P_{L_i} \left( u_i^{(k-1)} - u_i^{}  \right) \left[ \tilde {\lambda}_i^{(k)} N \tilde D_i^{(k)} - w_{{m+1}} \left( P_{W_t} - P_{W_{\Delta}} - N \tilde W_i^{(k)} \right) \right] \\
\end{aligned}
\end{equation}

The last term in (\ref{eqn:alpha_convergence}) can be bounded as
\begin{equation}
\label{eqn:last_term_of_u}
\begin{aligned}
& - \sum_{i=1}^{N}  2 \tau_k  \check P_{L_i} \left( u_i^{(k-1)} - u_i^{}  \right) \left[ \tilde {\lambda}_i^{(k)} N \tilde D_i^{(k)} - w_{{m+1}} \left( P_{W_t} - P_{W_{\Delta}} - N \tilde W_i^{(k)} \right) \right] \\
= & - 2 \tau_k \sum_{i=1}^{N}  \left( u_i^{(k-1)} - u_i^{}  \right) \left[ \hat {\lambda}_i^{(k-1)} N \hat D_i^{(k-1)} - w_{{m+1}} \left( P_{W_t} - P_{W_{\Delta}} - N \hat W_i^{(k-1)} \right) \right] \\
& - 2 \tau_k N \tilde {\lambda}_i^{(k)} \sum_{i=1}^{N}  \left( u_i^{(k-1)} - u_i^{} \right) \left( \tilde D_i^{(k)}  - \hat D_i^{(k-1)}  \right) \\
& - 2 \tau_k N \hat D_i^{(k)} \sum_{i=1}^{N}  \left( u_i^{(k-1)} - u_i^{}  \right) \left( \tilde {\lambda}_i^{(k)}  - \hat {\lambda}_i^{(k-1)} \right) \\
& - 2 \tau_k N w_{{m+1}} \sum_{i=1}^{N}  \left( u_i^{(k-1)} - u_i^{}  \right) \left( \tilde W_i^{(k)} - \hat W_i^{(k-1)} \right) \\
\leqslant & - 2 \tau_k \sum_{i=1}^{N}  \left( u_i^{(k-1)} - u_i^{}  \right) \mathcal{L}_{u_i} \left( \bm{u}^{(k-1)}, \hat {\lambda}_i^{(k-1)} \right) + 2 \tau_k N^2 D_{\lambda} \left \Vert \tilde D_i^{(k)} - \hat D_i^{(k-1)} \right \Vert  \\
& + 2 \tau_k N \Delta {\tilde P} \left \Vert \tilde {\lambda}_i^{(k)} - \hat {\lambda}_i^{(k-1)} \right \Vert + 2 \tau_k N^2 w_{{G}} \left \Vert \tilde W_i^{(k)} - \hat W_i^{(k-1)} \right \Vert \\
\leqslant &  2 \tau_k   \left( \Big \Vert \mathcal{L} \left( \bm{u}^{(k-1)}, \hat {\lambda}_i^{(k-1)} \right) \Big \Vert^2  + \Big \Vert \mathcal{L} \left( \bm{u}^{}, \hat {\lambda}_i^{(k-1)} \right) \Big \Vert^2 \right) + 2 \tau_k N^2 D_{\lambda} \left \Vert \tilde D_i^{(k)} - \hat D_i^{(k-1)} \right \Vert  \\
& + 2 \tau_k N \Delta {\tilde P} \left \Vert \tilde {\lambda}_i^{(k)} - \hat {\lambda}_i^{(k-1)} \right \Vert + 2 \tau_k N^2 w_{{G}} \left \Vert \tilde W_i^{(k)} - \hat W_i^{(k-1)} \right \Vert \\
\end{aligned}
\end{equation}

\begin{equation}
\label{eqn:lambda_convergence}
\begin{aligned}
& \sum_{i=1}^{N} \Big \Vert {\lambda}_{i}^{(k)} - {\lambda}^{} \Big \Vert ^2 \\
= & \sum_{i=1}^{N} \Big \Vert \left( \tilde {\lambda}_i^{(k)} + \tau_k \Big( \left(N \tilde D_i^{(k)} \right)^2 -\varepsilon \Big) \right)  -  {\lambda}^{} \Big \Vert ^2 \\
\leqslant &  \sum_{i=1}^{N} \Big \Vert { \tilde \lambda}_i^{(k)} - \lambda^{} \Big \Vert ^2 + \tau_k^2 {N} \Delta {\tilde P}^4 + \sum_{i=1}^{N}  2 \tau_k \left( {\tilde \lambda}_i^{(k)} - \lambda^{} \right) \bigg( \left(N \tilde D_i^{(k)} \right)^2 -\varepsilon \bigg) \\
\end{aligned}
\end{equation}

The last term in (\ref{eqn:alpha_convergence}) can be bounded as
\begin{equation}
\label{eqn:beta_convergence}
\begin{aligned}
&  \sum_{i=1}^{N}  2 \tau_k \left( {\tilde \lambda}_i^{(k)} - \lambda^{} \right) \bigg( \left(N \tilde D_i^{(k)} \right)^2 -\varepsilon \bigg) \\
= &  \sum_{i=1}^{N}  2 \tau_k \left( {\hat \lambda}_i^{(k-1)} - \lambda^{} + {\tilde \lambda}_i^{(k)} -{\hat \lambda}_i^{(k-1)} \right) \bigg( \left(N \tilde D_i^{(k)} \right)^2 -\varepsilon \bigg) \\
= &  \sum_{i=1}^{N}  2 \tau_k \left( {\hat \lambda}_i^{(k-1)} - \lambda^{} \right) \bigg( \left(N \hat D_i^{(k-1)} \right)^2 -\varepsilon \bigg) + \sum_{i=1}^{N}  2 \tau_k \left( {\hat \lambda}_i^{(k-1)} - \lambda^{} \right) \bigg( \left(N \tilde D_i^{(k)} \right)^2 - \left(N \hat D_i^{(k-1)} \right)^2 \bigg) \\
& + \sum_{i=1}^{N}  2 \tau_k \left( {\tilde \lambda}_i^{(k)} - {\hat \lambda}_i^{(k-1)} \right) \bigg( \left(N \tilde D_i^{(k)} \right)^2 -\varepsilon \bigg) \\
\leqslant &  \sum_{i=1}^{N}  2 \tau_k \left( {\hat \lambda}_i^{(k-1)} - \lambda^{} \right) \bigg( \left(N \hat D_i^{(k-1)} \right)^2 -\varepsilon \bigg) + \sum_{i=1}^{N}  2 \tau_k D_{\lambda} N^2 \left \Vert \tilde D_i^{(k)} - \hat D_i^{(k-1)} \right \Vert \left \Vert \tilde D_i^{(k)} + \hat D_i^{(k-1)} \right \Vert \\
& + 2 \tau_k \Delta {\tilde P}^2 \sum_{i=1}^{N}  \left \Vert \tilde {\lambda}_i^{(k)} - \hat \lambda^{(k-1)} \right \Vert \\
\leqslant & 2 \tau_k  \sum_{i=1}^{N} \left( \Big \Vert \mathcal{L} \left( \hat {\bm{u}}^{(k-1)}, \hat {\lambda}_i^{(k-1)} \right)  \Big \Vert^2 + \Big \Vert \mathcal{L} \left( \hat {\bm{u}}^{(k-1)},  {\lambda}^{} \right) \Big \Vert^2 \right) + 4 \tau_k D_{\lambda} N \Delta \tilde P \left \Vert \tilde D_i^{(k)} - \hat D_i^{(k-1)} \right \Vert \\
& + 2 \tau_k N \Delta {\tilde P}^2 \left \Vert \tilde {\lambda}_i^{(k)} - \hat \lambda^{(k-1)} \right \Vert
\end{aligned}
\end{equation}
where $\hat {\lambda}^{(k-1)} \leqslant C_{\lambda} $, $ \lVert \tilde D_i^{(k)} \rVert \leqslant \Delta \tilde P/N $, and $ \lVert \hat D_i^{(k-1)} \rVert \leqslant \Delta \tilde P/N $.
Thus, from the result of (\ref{eqn:aulxilary_variable_convergence}), $ {\bm{u}}^{(k)}$ and $\lambda^{(k)}$ converge to the common points $\hat {\bm{u}}^{(k)}$ and $\hat \lambda^{(k)}$, respectively.

We let $\bm{u}^{*}$ and ${\lambda }^{*}$ represent the saddle point. By combining (\ref{eqn:alpha_convergence})-(\ref{eqn:beta_convergence}), we obtain the inequality as follow
\begin{equation}
\label{eqn:total_primal_dual1}
\begin{aligned}
& \sum _{i=1}^{N} \Big \Vert {u}_i^{(k)} - {u}_i^{*} \Big \Vert^{2} \leqslant \sum _{i=1}^{N} \Big \Vert {u}_i^{(k-1)}-{u}_i^{*} \Big \Vert^{2} + \tilde c_k + 2 \tau_k  \sum _{i=1}^{N}   \left( \Big \Vert \mathcal{L} \left( \bm{u}^{(k-1)}, \hat {\lambda}_i^{(k-1)} \right) \Big \Vert^2  + \Big \Vert \mathcal{L} \left( \bm{u}^{}, \hat {\lambda}_i^{(k-1)} \right) \Big \Vert^2 \right),
\end{aligned}
\end{equation}
where
\begin{equation}
\begin{aligned}
\tilde c_k =  & \  4 \tau_k^2 N^2 \check P_{L_i}^2 \left( \Delta \tilde P + C_{\lambda} w_{G} P_{W_\Delta} \right)^2 + 2 \tau_k N^2 D_{\lambda} \left \Vert \tilde D_i^{(k)} - \hat D_i^{(k-1)} \right \Vert  \\
& + 2 \tau_k N \Delta {\tilde P} \left \Vert \tilde {\lambda}_i^{(k)} - \hat {\lambda}_i^{(k-1)} \right \Vert + 2 \tau_k N^2 w_{{G}} \left \Vert \tilde W_i^{(k)} - \hat W_i^{(k-1)} \right \Vert \\
\end{aligned}
\end{equation}

\begin{equation}
\label{eqn:total_primal_dual2}
\begin{aligned}
& \sum _{i=1}^{N} \Big \Vert {\lambda_i }^{(k)} - {\lambda }^{*} \Big \Vert^2 \leqslant \sum _{i=1}^{N} \Big \Vert {\lambda_i }^{(k-1)} - {\lambda }^{*} \Big \Vert^2 + \tilde c_k + 2 \tau_k  \sum_{i=1}^{N} \left( \Big \Vert \mathcal{L} \left( \hat {\bm{u}}^{(k-1)}, \hat {\lambda}_i^{(k-1)} \right)  \Big \Vert^{2} + \Big \Vert \mathcal{L} \left( \hat {\bm{u}}^{(k-1)},  {\lambda}^{} \right) \Big \Vert^{2} \right) ,
\end{aligned}
\end{equation}
where
\begin{equation}
\begin{aligned}
\tilde c_k =  & \ \tau_k^2 {N} \Delta {\tilde P}^4 + 4 \tau_k D_{\lambda} N \Delta \tilde P \left \Vert \tilde D_i^{(k)} - \hat D_i^{(k-1)} \right \Vert \\
& + 2 \tau_k N \Delta {\tilde P}^2 \left \Vert \tilde {\lambda}_i^{(k)} - \hat \lambda^{(k-1)} \right \Vert
\end{aligned}
\end{equation}

Since the $\lim_{k \rightarrow + \infty } \tau_k = 0 $, the last two terms on the right side of (\ref{eqn:total_primal_dual1}) and (\ref{eqn:total_primal_dual2}) converge to zeros as $k\rightarrow \infty$. That can verify that $ \lim_{k \rightarrow + \infty } \sum _{i=1}^{N} \Big \Vert {u}_i^{(k)} - {u}_i^{*} \Big \Vert^{2} $ exists for any $u \in \mathcal U$. It follows from \cite{distributed2012zhu}  that  $ \lim_{k \rightarrow + \infty } \sum _{i=1}^{N} \Big \Vert {u}_i^{(k)} - {u}_i^{*} \Big \Vert^{2} = 0 $. Similarly, it can obtain that $  \lim_{k \rightarrow + \infty } \sum _{i=1}^{N} \Big \Vert {\lambda}_i^{(k)} - {\lambda}_i^{*} \Big \Vert^{2} = 0 $.
Finally, the sequence $ ( \lVert \bm{u}^{(k)}-\bm{u}^{*} \rVert^{2} + \sum _{i=1}^{N} \lVert {\lambda }^{(k)} - {\lambda }^{*} \rVert^2  )$ converges for the saddle point $(\bm{u}^{*}, \lambda^*)$.

In the process of load shedding, the generation compensation is employed to reduce the power deficit $\Delta \tilde P$. Obviously it will improve the convergence of load shedding process. After the generation compensation, the saddle point is changed to $ (\bar {\bm{u}}^{*}, \bar{\lambda}^*)$. The saddle point satisfies that $ \lVert \bm{u}^{(0)} - \bar {\bm{u}}^{*} \rVert < \lVert \bm{u}^{(0)} - {\bm{u}}^{*} \rVert $ and $ \lVert \lambda^{(0)} - \bar {\lambda}^{*} \rVert < \lVert \lambda^{(0)} - {\lambda}^{*} \rVert $. It can be obtained that
\begin{subequations}
\begin{align}
\label{eqn:compensation_optimal_solution}
& \left( \lVert \bm{u}^{(k)}- \bar{\bm{u}}^{*} \rVert^{2} + \sum _{i=1}^{N} \lVert {\lambda }^{(k)} - \bar{\lambda }^{*} \rVert^2  \right) \\
& \quad < \left( \lVert \bm{u}^{(k)}-\bm{u}^{*} \rVert^{2} + \sum _{i=1}^{N} \lVert {\lambda }^{(k)} - {\lambda }^{*} \rVert^2  \right) .
\end{align}
\end{subequations}

It can be seen that (\ref{eqn:compensation_optimal_solution}) also satisfies the Lemma 1. Thus, $\bm{u}^{(k)} $ also converges to the $\bar {\bm{u}}^{*}$, and the supply-demand balance of the microgrid system can be achieved by the proposed DLSS method.


\end{document}